\documentclass[11pt]{article}
\usepackage{amssymb,amsmath,amsthm,enumerate,subdepth,upgreek,float}
\usepackage{fullpage}
\usepackage{xcolor}
\usepackage{graphicx}
\usepackage[colorlinks]{hyperref}
\hypersetup{linkcolor=[rgb]{0,0.2,0.6}}
\hypersetup{citecolor=[rgb]{0,0.6,0.2}}
\hypersetup{urlcolor=[rgb]{.8,0,.8}}

\usepackage[noline, boxed, noend, linesnumbered]{algorithm2e}
\SetAlCapSkip{5pt}
\SetArgSty{textnormal}

\usepackage[style=alphabetic, sorting=nyt, minalphanames=3, maxbibnames=7]{biblatex}
\usepackage{authblk}
\usepackage[titletoc,title]{appendix}
\usepackage[nameinlink]{cleveref}
\crefname{equation}{}{}

\newtheorem{theorem}{Theorem}[section]
\newtheorem{lemma}[theorem]{Lemma}

\newtheorem{definition}[theorem]{Definition}
\newtheorem{corollary}[theorem]{Corollary}
\newtheorem{proposition}[theorem]{Proposition}
\theoremstyle{definition}
\newtheorem*{remark}{Remark}

\let\originalleft\left
\let\originalright\right
\renewcommand{\left}{\mathopen{}\mathclose\bgroup\originalleft}
\renewcommand{\right}{\aftergroup\egroup\originalright}

\newcommand{\eps}{\varepsilon}
\newcommand{\mc}{\mathcal}

\newcommand{\Var}{\mathrm{\bf Var}}
\newcommand{\E}{\mathop{\mathbf{E}\mbox{}}\limits}
\newcommand{\R}{\mathbb{R}}
\newcommand{\C}{\mathbb{C}}

\newcommand{\dd}{\mathrm{d}}
\newcommand{\ct}{^\dagger}
\newcommand{\id}{\mathbb{I}}

\newcommand{\ugrp}{\mathbb{U}}
\newcommand{\haar}{\mathcal{U}}
\newcommand{\del}{\Updelta}
\newcommand{\phic}{\Upphi_{\mc C}}

\newcommand{\poly}{\mathrm{poly}}
\newcommand{\polylog}{\mathrm{polylog}}
\newcommand{\supp}{\mathop{\mathrm{supp}}}

\newcommand{\ket}[1]{|#1\rangle}
\newcommand{\bra}[1]{\langle#1|}

\newcommand{\Tr}{\mathrm{Tr}}

\newcommand{\spn}{\mathrm{span}}
\newcommand{\TrN}[1]{\left\|#1\right\|_1}

\newcommand{\DmN}[1]{\left\|#1\right\|_{\diamond}}
\newcommand{\FrN}[1]{\left\|#1\right\|_{\mathrm F}}

\title{Anti-Concentration for the Unitary Haar Measure and Applications to Random Quantum Circuits}

\author[1]{\normalsize{Bill~Fefferman}\thanks{wjf@uchicago.edu}}
\author[1]{Soumik~Ghosh\thanks{soumikghosh@uchicago.edu}} 
\author[2]{\normalsize{Wei~Zhan}\thanks{weizhan@purdue.edu}} 

\affil[1]{Department of Computer Science, The University of Chicago, Chicago, Illinois 60637, USA}
\affil[2]{Department of Computer Science, Purdue University, West Lafayette, Indiana 47907, USA}
\date{}

\begin{document}
	
	\maketitle
	
	\begin{abstract}
		We prove a Carbery-Wright style anti-concentration inequality for the unitary Haar measure, by showing that the probability of a polynomial in the entries of a random unitary falling into an $\eps$ range is at most a polynomial in $\eps$. Using it, we show that the scrambling speed of a random quantum circuit is lower bounded: Namely, every input qubit has an influence that is at least inverse exponential in depth, on any output qubit touched by its lightcone. Our result on scrambling speed works with high probability over the choice of a circuit from an ensemble, as opposed to just working in expectation. 
        
        As an application, we give the first polynomial-time algorithm for learning log-depth random quantum circuits with Haar random gates up to polynomially small diamond distance, given oracle access to the circuit. Other applications of this new scrambling speed lower bound include:     
		\begin{itemize}
			\item An optimal $\Omega(\log \eps^{-1})$ depth lower bound for $\eps$-approximate unitary designs on any circuit architecture;
			\item A polynomial-time quantum algorithm that computes the depth of a bounded-depth circuit, given oracle access to the circuit.
		\end{itemize}
        Our learning and depth-testing algorithms apply to architectures defined over any geometric dimension, and can be generalized to a wide class of architectures with good lightcone properties.
	\end{abstract}
	\newpage
	\section{Introduction}
	Random quantum circuits are one of the most popular paradigms of quantum computation in the near-term era. They are well-studied, both theoretically and experimentally, in the context of quantum advantage demonstrations: e.g., see \cite{Boixo_2018, Bouland18, Arute2019, Bouland22, Morvan23,   Movassagh23, Fefferman23}. They also have numerous applications in areas like benchmarking, like in \cite{Dankert09, Liu22}, in cryptography, as in e.g., \cite{Bassirian24, Aaronson23, Schuster24}, and in the modeling of physical objects like black holes, as in e.g., \cite{Hayden07,Piroli20, Yang23}.
	
One reason random quantum circuits are extensively studied is because they are rapid ``scramblers" of information. Intuitively, this means that the output state it generates has non-trivial correlations across spatially separated qubits. The rate of scrambling depends on the depth---the deeper the circuit is, the better it is at scrambling. However, even though previous works on scrambling have put \emph{upper bounds} on the scrambling speed of random circuits---see, e.g., \cite{Harrow09, Brown13, Nahum18, Harrow21, 
jian2022lineargrowthcircuitcomplexity,
chen2024efficientunitarydesignspseudorandom, chen2024efficientunitarytdesignsrandom, haah2024efficientapproximateunitarydesigns, metger2024simpleconstructionslineardepthtdesigns, Lieb72,Chen_2021,Wilming22,Chen_2023, PhysRevLett.124.180601}---\emph{lower bounds} on the scrambling speed are not well studied. 

    In this work, we give a \emph{lower bound} on the speed with which random quantum circuits can scramble information. Our main theorem is as follows:
    \begin{theorem}
    \label{thm:mixing}
		Let $\mathcal{C}$ be a random quantum circuit with a fixed architecture, where each gate is a $k$-qubit independent Haar random unitary. Let $\rho$ and $\pi$ be a pair of input and output qubits connected by $D$ layer of gates. Arbitrarily fix the inputs to $\mathcal{C}$, except the qubit $\rho$, and let $\Phi_{\mathcal{C}}$ be the channel that maps $\rho$ to $\pi$.
		
        Then for every $\gamma>0$, with probability at least $1-\gamma$ over $\mc C$ the following holds: For every two single-qubit states $\rho$ and $\rho'$, 
        \begin{equation}\label{eq:mixing}
            \FrN{\phic(\rho)-\phic(\rho')}\geq \FrN{\rho-\rho'}\cdot (2^{-D}\gamma)^{c_k}
        \end{equation}
		where $\FrN{\cdot}$ denotes the Frobenius norm, and $c_k > 0$ is a constant that depends only on $k$.
    \end{theorem}
    
    In particular, our main result gives a lower bound on the influence that changing an input qubit has on a designated output qubit inside the lightcone of that input qubit and shows that it decays at most exponentially fast with depth. We articulate some points that make \Cref{thm:mixing} useful for our applications:
    \begin{itemize}
        \item The bound in \Cref{thm:mixing} is uniform, in the sense that ratio $(2^{-D}\gamma)^{c_k}$ does not depend on how the input qubits are chosen, and is applicable to any circuit architecture. Moreover, the probability quantifier over the circuit $\mc C$ is stated before the choices of $\rho$ and $\rho'$, which is specifically important for our application on learning random quantum circuits (\Cref{thm:learncirc}).
        \item Although we stated the theorem using Frobenius norm, the norm is only over single-qubit states and thus is equivalent to any other matrix norm.
        \item To obtain the strongest bound, we can choose $D$ to be the shortest path from $\rho$ to $\pi$. When $\pi$ is outside the lightcone of $\rho$, we think of $D=\infty$, in which case the statement still holds.
        \item Our bound in \Cref{thm:mixing} is tight, in the sense that one could prove a matching upper bound 
        \begin{equation}\label{eq:mixingub}
            \FrN{\phic(\rho)-\phic(\rho')}\leq \FrN{\rho-\rho'}\cdot (2^{-D}\gamma)^{c_k'}
        \end{equation}
        for various architectures with a different constant $c_k'<c_k$, by calculating the moment $\E\left[\FrN{\phic(\rho)-\phic(\rho')}^2\right]$. Such calculation can be done by analyzing the Markov chain of Pauli operators, using the results in \cite{Harrow09,Brown13,Nahum17}. However, the moment method could not yield our lower bound, which we will explain in details below.
    \end{itemize}

    \paragraph{Implications on proving typicality statements.} Observe that \Cref{thm:mixing} is a statement that works \emph{with high probability} over the choice of circuit from the ensemble $\mathcal{C}$, as opposed to just \emph{in expectation} over the ensemble $\mathcal{C}$. To the best of our knowledge, this is the first typicality result regarding the scrambling behavior of random quantum circuits.

    Most of the previous work studying random quantum circuits examines their properties by computing some specific quantities of interest in expectation, such as entanglement entropy, OTOC and collision chance, by relating the quantities with the lower moments (usually the second moment) of the quantum circuit. For non-negative quantities the expectation already provides an upper bound that holds with high-probability (like in \eqref{eq:mixingub}) by Markov's inequality. But for proving a strong typicality result as in \eqref{eq:mixing}, one also needs a corresponding lower bound, which is commonly proved by calculating the variance of the desired quantity $X$. In particular, we would want that for some appropriately small $\epsilon$, and for some $t\geq 1$,
    \begin{equation}
    \label{typicality}
    \E[ X^{2t}] \leq (1+\epsilon) \E[ X^t]^2.
    \end{equation}
    However, for most architectures a relation like \eqref{typicality} is largely unknown. One reason is that when $X$ is already a second moment of the quantum circuit, and proving \eqref{typicality} requires computing the fourth or higher moments, which is extremely difficult even for standard brickwork circuits \cite{Braccia24}. More importantly, in our case, $X$ can be heuristically approximated by the product $\lambda_1\cdots\lambda_D$ of i.i.d random variables $\lambda_1,\ldots,\lambda_D$ that follow some distribution $\Lambda$ with constant variance; because of this, we cannot expect \eqref{typicality} to hold as $\E[ X^{2t}]$ will be exponentially larger than $\E[ X^t]^2$.
    

    In this work, we overcome this barrier by performing a more fine-grained analysis. We do not analyze the randomness present in the choice of gates all at once; Instead, we break up the randomness of the circuit into $D$ layers and lower bound the random variable $X$ with a product of $D$   random variables that are highly correlated. We then decouple these random variables using a new anti-concentration inequality on the unitary Haar measure, which will be stated in \Cref{thm:anticc}. This allows us to avoid calculating higher moments of the circuit, and still retrieve information about the distribution of $X$ that could be lost in the moments. Our result in \Cref{thm:mixing} directly implies typicality guarantees on other quantities, such as the purity of a single output qubit, and we believe the method we use can be generalized to prove typicality results for many other useful quantities.

    \subsection{Applications}
    \label{applications}
    
    Intuitively, \Cref{thm:mixing} means that the output state of random quantum circuits with logarithmic depth carries a signal pertaining to its input state that can be extracted from the output state using single qubit quantum state tomography. This allows us to show many  applications.

    \subsubsection{Optimal 2-design lower bounds} First, as an immediate corollary of \Cref{thm:mixing}, we show a depth lower bound for random quantum circuits, defined on any architecture, forming approximate unitary designs; see \Cref{sect:depthlb}.
	
	\begin{theorem}\label{thm:depthlb}
		Let $\mathcal{C}$ be a random quantum circuit, defined on an arbitrary fixed architecture of minimum depth $D$, where each gate is a $k$-qubit independent Haar random unitary. If $\mathcal{C}$ is an $\eps$-approximate $2$-design, then 
        \[
        D = \Omega_k(\log \eps^{-1}).
        \]
	\end{theorem}
	
	Here the minimum depth means the smallest number of gates that any path from input to output goes through, resulting in a stronger requirement than the previous depth lower bounds that uses the lightcone size property. \Cref{thm:depthlb} shows that previous works on approximate $t$-design constructions, e.g. \cite{Harrow09, hunterjones2019unitarydesignsstatisticalmechanics, Haferkamp22, Chen24Incomp, Schuster24}, are optimal in $\epsilon$, assuming the gates are Haar random unitaries. And combined with the $\Omega(\log n)$ depth lower bound in \cite{Dalzell22,Schuster24}, it means that the approximate $t$-design construction of depth $O(\log(n/\eps)\cdot t\mathop{\polylog}t)$ from \cite{Schuster24} is optimal in both $n$ and $\eps$.
 
    

    \subsubsection{Testing for depth}
    Our second application concerns testing the depth of a random circuit when the depth is promised to be at most logarithmic.
	
	\begin{theorem}\label{thm:depthtest}
		Let $\mathcal{C}$ be a brickwork random quantum circuit on $n$ qubits of an unknown depth $D$ promised to be bounded as $D=O(\log n)$, where each gate is independently Haar random. Given oracle access to $\mathcal{C}$, there is a polynomial time algorithm that outputs $D$ with probability $1-1/\poly(n)$.
	\end{theorem}
	
	Although \Cref{thm:depthtest} is stated with brickwork circuits for simplicity, it is applicable to many other architectures satisfying the following: For a random circuit of depth $D$, the lightcone size is strictly smaller than that of a random circuit of depth larger than $D$. Hence, there will be a particular output qubit which will be insensitive to a change in a particular input qubit in the former case and will be sensitive in the latter case. For $D=O(\log n)$ this difference is inverse polynomially large and can be tested by state tomography. The explicit algorithm is given in \Cref{sect:depthtest}.
 
    Note that our algorithm outputs the exact depth instead of obtaining an approximation, in contrast with the recently proposed depth test algorithm in \cite{Hangleiter24}. This is partially due to the fact that our algorithm works with high probability with respect to the choice of a circuit from the ensemble, as opposed to working only in expectation, unlike \cite{Hangleiter24}. 
    
    

	\subsubsection{Learning brickwork random circuits} 
    Third, we show that \Cref{thm:mixing} allows us to learn brickwork random circuits of logarithmic depth. We start by showing that the first layer of gates can be learned given oracle access to the circuit.
 
    \begin{theorem}\label{thm:learngate}
		Let $\mc C$ be a brickwork random quantum circuit on $n$ qubits of known depth $D=O(\log n)$, where each gate is independently Haar random. Given oracle access to $\mc C$, there is a polynomial time algorithm that with high probability outputs each gate in the first layer with polynomially small error.
	\end{theorem}

	Furthermore, in real life scenarios we can assume the distribution over the gates is a discrete approximation of the Haar measure (see \Cref{def:eps}), and in this case we can actually learn the entire circuit:
    \begin{theorem}
    \label{thm:learncirc}
		Let $\mathcal{C}$ be a brickwork random quantum circuit on $n$ qubits of known depth $D=O(\log n)$, where each gate is independently drawn from a discretized version of the Haar measure. Given oracle access to $\mathcal{C}$, there is a polynomial time algorithm that with high probability outputs $\mathcal{C}$ with polynomially small error.
	\end{theorem}
	
	We will prove \Cref{thm:learngate} and \Cref{thm:learncirc} in \Cref{sect:learn}. Note that the learning algorithm in \Cref{thm:learncirc} is proper, that the learnt circuit has the exact same depth and architecture as the actual circuit $\mathcal{C}$.
    
    Before this work, the state-of-art learning algorithm for brickwork quantum circuits is due to \cite{Huang24,Landau24}, which runs in polynomial time for circuits of $k$-dimensional geometry up to $O(\log^{1/(k+1)} n)$ depth, and is not proper (that the outputted circuit has depth larger than the input circuit by at least a constant factor). Our algorithm works for all geometric dimensions, works up to $O(\log n)$ depth, and still has polynomial runtime. There is evidence that the depth dependence of our learning algorithm is optimal, because random circuits form high enough designs beyond $\log n$ depth \cite{Schuster24}, and possibly also pseudorandom unitaries, which hints at learning hardness. 
 
    Here we provide an informal description of the learner:
    \begin{itemize}
    \item For each gate in the first layer, the learner iteratively searches through the gate set for the correct gate. There's only polynomially many gates in the gate set because we chose an appropriate $\eps$-net over $2$-qubit Haar random unitaries. During each round, the learner applies the inverse of a gate chosen from the gate set at the particular location it wishes to learn.
    \item Select an input qubit to the gate we are trying to learn. The key observation is that, depending on whether we applied the correct inverse, the output lightcone will consist of different qubits. In particular, one qubit will lie outside the lightcone and be unaffected by a change in the selected input qubit if we hit the correct inverse, as opposed to if we didn't. As a corollary of \Cref{thm:mixing}, in the former case, we can detect the influence of the input qubit by performing single qubit state tomography.
    \begin{figure}[h]
        \centering
        \begin{minipage}{.5\textwidth}
          \centering
          \includegraphics[page=1, width=\linewidth]{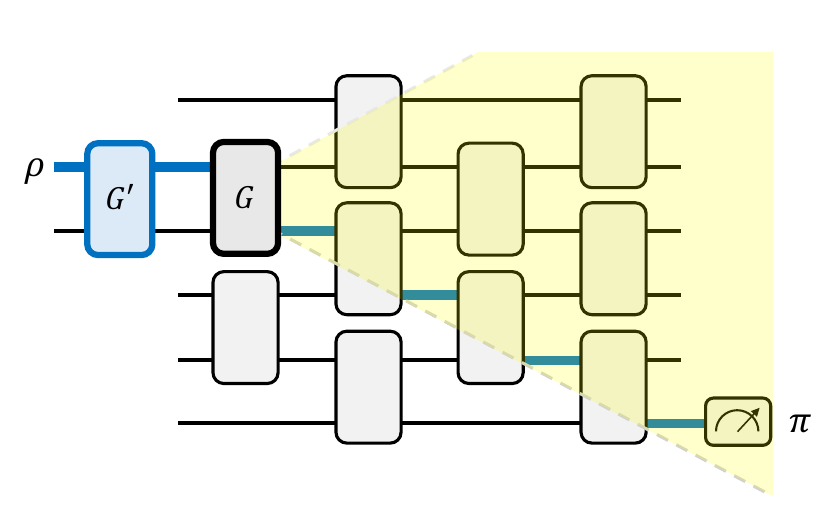}
        \end{minipage}%
        \begin{minipage}{.5\textwidth}
          \centering
          \includegraphics[page=2, width=\linewidth]{lc.pdf}
        \end{minipage}
        \caption{An illustration of the learning algorithm in \Cref{thm:learncirc}, where we try to uncompute the gate $G$ by applying some two-qubit unitary $G'$ before it. Left: When $G'$ failed to uncompute $G$, the lightcone of the input qubit $\rho$ touches an output qubit $\pi$ at distance $D$. Right: When $G'=G^\dagger$, the output qubit $\pi$ falls out of the lightcone and will not be affected by the change in $\rho$.}
        \label{fig:lightcone}
    \end{figure}
    \end{itemize}
    
    Note that the learner of \cite{Huang24,Landau24} uses a very different strategy: it looks at each output qubit, and then employs a brute force search over all possible lightcones of that circuit to find the correct one. Thereafter, it stitches together all the different local lightcones of every qubit, and then corrects for the errors in the overlapping regions of the different lightcones. The fact that we do not do a brute force search over all possible lightcones makes our algorithm efficient for any dimension, and not just $1$-dimensional brickwork circuits.
	
	\subsection{Main Technical Tool: Anti-Concentration for Haar Measure}
	
	The intuition behind the proof of \Cref{thm:mixing} is the following: We consider the path of $D+1$ qubits $\rho=\rho_0,\rho_1,\ldots,\rho_D=\phic(\rho)$ in the circuit $\mc C$, where gate $G_i$ has $\rho_i$ as an input and $\rho_{i+1}$ as an output. Let $\rho_i'$ be the corresponding qubits when the input is $\rho'$, and we want to bound the ratios $\lambda_i = \FrN{\rho_i-\rho'_i}/\FrN{\rho_{i-1}-\rho'_{i-1}}$ and hence their product.
	
	It turns out that we can prove the lower bound $\lambda_i\geq|F(G_i)|$, where $F$ is a polynomial function over the entries of $G_i$ in its matrix representation. Therefore, we only need to show that $|F(G_i)|$ is often not too small, when $G_i$ is a Haar random unitary. In other words, we need to show that the polynomial does not concentrate around zero. Our main technical contribution is the following theorem which proves this anti-concentration phenomenon:
	
	\begin{theorem}\label{thm:anticc}
		Let $U$ be a Haar random $n\times n$ unitary matrix, and let $F:\C^{2n^2}\to\C$ be a degree-$d$ polynomial on the entries of $U$ and $U\ct$. Then for every $\eps>0$, it holds that
		\[\Pr\left[\big|F(U,U\ct)\big|^2\leq\eps \E\!\big[|F(U,U\ct)|^2\big]\right]\leq C'(n,d)\cdot \eps^{C(n,d)}\]
		where $C(n,d)>0$ and $C'(n,d)>0$ are constants that depend only on $n$ and $d$.
	\end{theorem}
	
	The anti-concentration inequality of polynomials over Gaussian random variables was famously proved by Carbery and Wright \cite{Carbery01}, and their result actually applies to any log-concave distribution over $\R^n$. However, as the Haar measure does not even have a convex support, the proof techniques in \cite{Carbery01} does not apply. The alternative inductive proof in \cite{Lovett10} also fails to apply, as we are dealing with correlated input variables. We present a very different inductive proof in \Cref{sect:anticc}.
	
	Note that if we consider $F(U,U\ct)$ as a complex random variable, and define the complex variance
	\[\Var[F]=\E[|F|^2]-\left|\E[F]\right|^2=\min_{z\in\C}\E[|F-z|^2],\]
	then we obtain the form closer to the classical anti-concentration inequalities:
	
	\begin{corollary}\label{col:anticc}
		Let $U$ be a Haar random $n\times n$ unitary matrix, and let $F$ be a degree-$d$ polynomial on the entries of $U$ and $U\ct$. Then for every $\eps>0$ and every $z\in\C$, it holds that
		\[\Pr\left[|F-z|^2\leq\eps\Var[F]\right]\leq C'(n,d)\cdot \eps^{C(n,d)}.\]
	\end{corollary}
	
	However in this work we will not use the form in \Cref{col:anticc}, as \Cref{thm:anticc} suffices for our applications.
	
	\begin{remark}
		Unlike the Carbery-Wright inequality which is dimension-free, meaning the right hand side is $C'(d)\cdot\eps^{C(d)}$ and does not depend on $n$, we have $C(n,d)=(4n^2d)^{-1}$ and $C'(n,d)=O(n^3d)$ in our proof. In fact, using the concentration bounds it is not hard to show that $C'(n,d)$ must depends polynomially on $n$. However, we conjecture that $C(n,d)$ could be independent of $n$, in which case the result would be applicable to random quantum circuits with gates of higher locality.
	\end{remark}
 
	\subsection{Related Works}
	
	Concentration phenomenon on unitary Haar measures has been extensively studied, and the readers can refer to \cite{Meckes19} for a comprehensive review of the results. In comparison, much less has been shown for the reverse direction, namely the anti-concentration inequalities.
	
	One common way to prove such inequalities is by calculating higher moments and apply the Paley-Zygmund inequality, which was indeed used for showing the anti-concentration property of \emph{output distributions} of Haar random unitaries and random quantum circuits \cite{Aaronson11,Hangleiter18,Dalzell22}. However, the inequality proved this way is not strong enough for our applications, while calculating moments of a random quantum circuit is also non-trivial and depends highly on the architecture \cite{Fisher23,Braccia24}. Instead, we resort to prove a general anti-concentration inequality for polynomials, whose theory has been well developed for Gaussian distributions \cite{Carbery01} and product distributions (namely the Littlewood-Offord theory) and has found numerous applications in computational complexity theory \cite{Meka13,Meka16,Kane17}. Our inductive proof of \Cref{thm:anticc} also shares a similar spirit with the elementary proof of Carbery-Wright inequality in \cite{Lovett10}.
	
	Multiple notions of scrambling property of random quantum circuits has been previously studied. In particular, \cite{Brown13} showed that in a random circuit consists of $O(n\log^2 n)$ sequential applications of Haar random gates on a complete graph of $n$ qubits, every subset of $cn$ qubits is polynomially close to maximally mixed with high probability for some constant $c>0$. Our \Cref{thm:mixing} can be viewed as a result in the reverse direction which bounds the scrambling speed of such random circuits. Specifically, at least $\Omega(n\log n\log\log n)$ sequential gates are required, as otherwise with high probability a pair of input and output qubits are $o(\log n)$ depth apart due to a generalization of the coupon-collector problem \cite{Erdos61}. It is also reasonable to believe that our method of proving \Cref{thm:mixing}, via the anti-concentration inequality, is applicable to obtain lower bounds for other measures of scrambling such as entanglement and out-of-time-ordered correlation (OTOC) \cite{Nahum17, Nahum18, Bertini20, Harrow21}. Upper bounds in the above-mentioned works are obtained by calculating moments and analyzing the averaged Markov chain on Pauli operators, which are not sufficient to prove lower bounds in the typical case.
	
	We also review some previous works related to our applications and clarify the connections. For approximate unitary designs, many previous constructions, for example \cite{Harrow09,Brandao12,Haferkamp22,Harrow23,Chen24Incomp}, employed Haar random unitary gates and achieved the optimal $O(\log \eps^{-1})$ dependence on $\eps$. The construction of depth $O(\log(n/\eps)\cdot t\mathop{\polylog}t)$ in the recent work of \cite{Schuster24} also falls into this category. Meanwhile, they also proposed a construction of approximate 3-design with only $O(\log\log(n/\eps))$ depth. This does not contradict our lower bound \Cref{thm:depthlb} as the construction uses random Clifford unitaries, which are not independent between layers and also fail the anti-concentration property in \Cref{thm:anticc}. It is intriguing, however, to see if our argument can be extended to show a matching $\Omega(\log\log \eps^{-1})$ depth lower bound.
	
	The depth test algorithm in \cite{Hangleiter24} is based on the entanglement dynamics of random circuits and implemented with their Bell sampling framework. For brickwork circuits, as there is a constant gap between the upper and lower bounds for the entanglement entropy in the typical case, their algorithm gives constant approximation of the depth. For general architectures, the algorithm in \cite{Hangleiter24} also requires knowledge of the entanglement velocity, while our algorithm for \Cref{thm:depthtest} only relies on a specific property of the architecture that the lightcone strictly expands with depth.
	
	The learning algorithm for shallow quantum circuits in \cite{Huang24} is based on the idea of brute-force enumerating all possibilities in a light cone, and stitching the parts together. Therefore, their algorithm works any quantum circuit in worse case, and has complexity exponential in the lightcone size. This means that in order to have polynomial efficiency, the depth has to be $O(\log^{1/k} n)$ for $k$-dimensional geometrically local circuits and $O(\log\log n)$ for general architectures. However, the learning algorithm is improper and the outputted circuit has a large polynomial blow-up in depth. The blow-up factor could be reduced to constant, at the cost of lowering the depth bound to $O(\log^{1/{(k+1)}} n)$, and thus could not perform in polynomial time for log-depth circuits even for 1-dimensional geometry. In comparison, our algorithm works for random quantum circuits, on the fixed brickwork or similar geometrically local architecture, where neighboring qubits can be distinguished by their lightcones. The strength of our algorithm \Cref{thm:learncirc} is that it works up to logarithmic depth in polynomial time regardless of the geometric dimension, and it learns properly in that it outputs a circuit with the exact same architecture. This makes our result incomparable with \cite{Huang24}.
	
	We also mention that, there is a different learning task where instead given oracle access to the circuit $\mc C$, the learner is only given copies of the state $\mc C\ket{0^n}$, and needs to learn a circuit that prepares the same state with error in trace distance. Our algorithm relies on having different inputs and hence does not work in this case, whereas \cite{Huang24} gave a quasi-polynomial efficiency algorithm for 2D and the algorithm was extended to higher dimensions in \cite{Landau24}.

    \section{Technical Overview}
    In this section we provide some intuitions for the technical parts in our proofs, as our approach significantly differs from previous techniques for proving Carbery-Wright type bounds and proving properties of random quantum circuits.

    \subsection{Induction for Anti-Concentration Bound}
    We start with a bare-bone inductive proof for \Cref{thm:anticc}, but instead of having the input to $F$ being the entries of a Haar random unitary, we think of them as independent random variables $x_1,\ldots,x_n$. In this case, we can write $F$ as
    \begin{equation}\label{eq:poly}
        F(x_1,\ldots,x_n)=x_1^d F_d(x_2,\ldots,x_n)+x_1^{d-1}F_{d-1}(x_2,\ldots,x_n)+\cdots+F_0(x_2,\ldots,x_n),
    \end{equation}
    where each $F_i$ is a polynomial of degree at most $d$. As a result, we can define a single-variable polynomial $P(x_1)$ as
    \begin{align*}
        P(x_1) &=\E_{x_2,\ldots,x_n}[F(x_1,\ldots,x_n)] \\
        &=\E[F_d]\cdot x_1^d+\E[F_{d-1}]\cdot x_1^{d-1}+\cdots+\E[F_0],
    \end{align*}
    and prove the anti-concentration property of $P(x_1)$ using analytic methods. Notice that for each fixed $x_1$ such that $P(x_1)$ is not concentrated, $F$ is also a polynomial $F'_{x_1}$ over $x_2,\ldots,x_n$, and therefore reduce the problem to proving anti-concentration for $F'_{x_1}$, as by the union bound we have
    \begin{align*}
        &\ \Pr\big[F(x_1,\ldots,x_n)\leq\eps\E[F]\big] \\
        \leq&\ \Pr\big[P(x_1)\leq\sqrt{\eps}\E[P]\big] +\Pr\big[F'_{x_1}(x_2,\ldots,x_n)\leq\sqrt{\eps}\E[F'_{x_1}]\big].
    \end{align*}
    
    The probability obtained via this induction will be dependent on $n$ in the exponent of $\eps$. In order to obtain a dimension-free bound like in Carbery-Wright, \cite{Lovett10} first converts $F$ into an equivalent polynomial that is multi-linear over a partition of input variables, and by taking $x_1$ to be one part of the partition instead of a single variable, $F_0$ in \eqref{eq:poly} could be made to have degree at most $d-1$. This way it is guaranteed that the degree of $F'_{x_1}$ is strictly smaller than that of $F$, and thus the induction could work on $d$ instead of $n$.

    Unfortunately, the conversion of polynomial in \cite{Lovett10} does not work for the case when input variables are not independent. In fact, with correlated input we are facing a more fundamental problem: it is not even guaranteed that $P(x_1)$ is a polynomial function, as $\E\nolimits_{x_2,\ldots,x_n}[F_i]$ is not a constant anymore, but a function in $x_1$ itself. A crucial step in our proof is to show that, with some proper grouping of the entries in the Haar random matrix, each $\E[F_i]$ is indeed a polynomial function in $x_1$ of degree at most $d-i$, and hence the induction proof goes through, albeit with dependence on $n$; See \Cref{sect:anticc}.

    \subsection{Tracking 3-Dimensional Pauli Subspaces}
    The main objects we care about in \Cref{thm:mixing} are the differences between density matrices of states, which are linear combinations of non-identity Pauli operators. In particular, the difference on a single qubit is spanned by $\{X,Y,Z\}\otimes I^{\otimes(n-1)}$, where the non-identity Pauli $X,Y$ or $Z$ acts on the qubit where the difference is taken. In this language, \Cref{thm:mixing} in essence claims the following: The subspace of Pauli operators $\{X,Y,Z\}_0\otimes I^{\otimes(n-1)}$ indicating the differences in the input qubit $\rho_0$, when evolved by the circuit $\mc C$ (still under the Schr\"{o}dinger picture), will have non-negligible projection onto the subspace $\{X,Y,Z\}_D\otimes I^{\otimes(n-1)}$ which indicates the differences in the output qubit $\rho_D$.

    \begin{figure}[h]
        \centering
        \includegraphics[width = .85\textwidth]{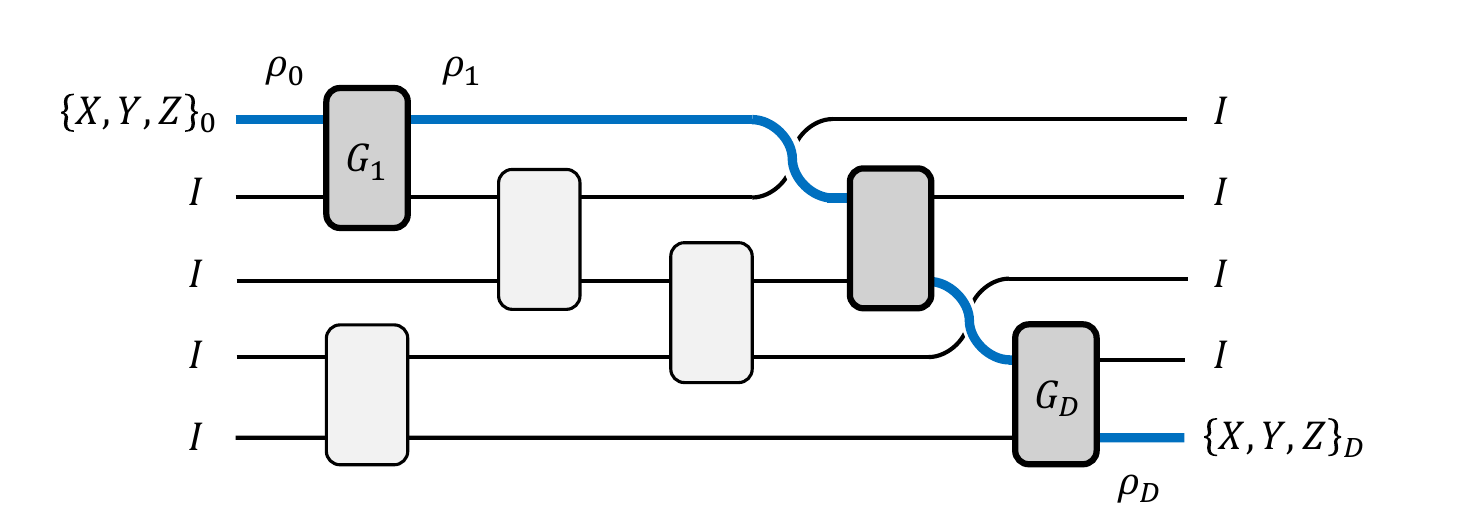}
        \caption{A Illustration of the claim and proof of \Cref{thm:mixing} with an example circuit, where the bold wires represents a path from the first input qubit to the last output qubit. The state of the qubit on the $i$-th segment of the wire, between gate $G_{i-1}$ and $G_i$, is denoted as $\rho_i$, whose difference is captured by the projection onto the Pauli subspace $\{X,Y,Z\}_i\otimes I^{\otimes(n-1)}$.}
    \end{figure}
    
    Our proof of \Cref{thm:mixing} keeps track of the evolution of the 3-dimensional Pauli subspace:
    \[\mc S_i=\mc C_i\cdot\left( \mathrm{span}\{X,Y,Z\}_0\otimes I^{\otimes(n-1)}\right)\cdot\mc C_i\ct,\]
    where $\mc C_i$ is the part of circuit $\mc C$ up until the $i$-th layer. Denoting the qubits on the length-$D$ path from the input to the output as $\rho_0,\rho_1,\ldots,\rho_D$, the projection of $\mc S_i$ onto $\{X,Y,Z\}_i\otimes I^{\otimes(n-1)}$ contains all possible differences $\rho_i-\rho_i'$, and the goal is to bound the singular values of this projection. Notice that this projection is not affected by any unitary gates in the $i$-th layer, other than the one $G_i$ on the path that connects $\rho_{i-1}$ and $\rho_i$. As a result, we can lower bound the ratio
    \[\frac{\FrN{\rho_i-\rho_i'}}{\FrN{\rho_{i-1}-\rho_{i-1}'}}\]
    with the smallest singular value of a linear transformation on $\mc S_{i-1}$ defined by $G_i$, and in turn with a polynomial function $F_i(G_i)$ on the entries of $G_i$, whose coefficients depend on $\mc S_{i-1}$.

    Now, to show that each $F_i(G_i)$ with high probability cannot be too small, we can use exactly our anti-concentration bound \Cref{thm:anticc} to argue that $F_i(G_i)$ cannot be concentrated around zero as $G_i$ is Haar random. However, to bound their product $F_1(G_1)\cdots F_D(G_D)$, we cannot treat them individually as the polynomial functions themselves are correlated. Instead, we use \Cref{thm:anticc} to find an anti-concentrated distribution $\Lambda_i$ that uniformly lower bounds the distribution of $F_i(G_i)$, regardless of the choice of $\mc S_{i-1}$. This way, we can decouple the distributions of the multipliers, and prove a lower bound on the independent product $\Lambda_1\cdots\Lambda_D$; See \Cref{sect:mixing} for details.
    
    \subsection{Learning Log-Depth Random Quantum Circuits}
    We already presented the general idea of the learning algorithm in \Cref{applications}. However, that is far from the whole story: one major problem we need to solve is that, while our soundness assurance \Cref{thm:mixing} is an average-case statement over the distribution of random circuits, learning each gate correctly is a worst-case statement. In other words, \Cref{thm:mixing} does not directly guarantee that for \emph{every} attempt that fails to uncompute the current gate, the input qubit will have an observable influence on the output qubit.

    In fact, such a worst-case statement is impossible with our approach: the lightcone property cannot tell the difference when the gate $G$ is perfectly uncomputed by $G'=G^\dagger$, or when $G'=(U_1\otimes U_2)\cdot G^\dagger$ for some non-identity single-qubit unitaries $U_1$ and $U_2$. Fortunately it is unnecessary to distinguish this two cases: In the latter case, $G'^\dagger\cdot G=U_1^\dagger\otimes U_2^\dagger$, as the residue of the learning process for the current layer, can be viewed as part of the next layer by composing $U_1^\dagger$ with one of the next gates and $U_2^\dagger$ with the other. This does not affect the distribution of the remaining circuit as each gate is still independently Haar random, which is where we actually apply \Cref{thm:mixing} on.
    
    Meanwhile, the actual worst-case statement for correctness is the following: when changing the first input qubit to the composite gate $G'^\dagger\cdot G$, the second output qubit will be affected if and only if $G'$ cannot be written as $(U_1\otimes U_2)\cdot G^\dagger$. This statement is not hard to prove by itself, but we need a robust version of it in both completeness and soundness to account for learning errors, as we could only go over an $\eps$-net of gates and state tomography at the end also introduces error. These two directions are handled in \Cref{lemma:complete} and \Cref{lemma:sound} respectively; See \Cref{sect:app} for the detailed proof, and other more straightforward applications.
    
	\section{Preliminaries}
	
	We start with some basic notations. We use $\ugrp(d)$ to denote the unitary group of dimension $d$, and use $\haar(d)$ to denote the Haar measure over $\ugrp(d)$. We use Greek letters such as $\rho,\pi,\tau$ to denote density matrices of quantum states. We use $\TrN{\cdot}$ and $\FrN{\cdot}$ for trace norm and Frobenius norm, and $d_\diamond(\cdot,\cdot)$ for diamond distance between unitary channels, defined as
    \[d_\diamond(\Phi_1,\Phi_2)=\max_\rho\TrN{(\Phi_1\otimes\id)\rho-(\Phi_2\otimes\id)\rho},\]
    with maximum taken over all density matrices $\rho$.
	
	A circuit architecture determines the positions of gates in the circuit. We define the depth of an architecture as follows.
	
	\begin{definition}\label{def:depth}
		In a circuit architecture, a path in space-time between two qubits $\rho$ and $\rho'$ is a sequence of qubits $\rho=\rho_0,\rho_1,\ldots,\rho_D=\rho'$ where for each $i$, there is a gate in the circuit that has $\rho_i$ as an input and $\rho_{i+1}$ as an output.
        
        We say the architecture has minimum depth $D$, if there exists a pair of input and output qubits connected by a path of length $D$, and every other path from input to output has depth at least $D$.
	\end{definition}
	
	A specific architecture of interest is the (1-dimensional) brickwork architecture:
	
	\begin{definition}
		A brickwork quantum circuit on $n$ qubits of depth $D$ consists of $D$ layers of gates, where on layer $j$, there is a two-qubit gate $G_{i,j}=G_{i+1,j}$ acting on the $i$-th and $(i+1)$-th qubit if and only if $i$ and $j$ have the same parity.
	\end{definition}
	
	The brickwork architecture could be generalized to higher dimensional geometry, and our results still hold for any constant dimension. However, for simplicity we stick with the 1-dimensional architecture in this paper.
	
	We will need the following statements about quantum state tomography and quantum process tomography on single qubits for our algorithms (see e.g. \cite{Nielsen_Chuang_2010})
	
	\begin{proposition}\label{prop:tomo}
		Given access to copies of a single-qubit state $\rho$, one can output an estimation $\tilde{\rho}$ with $\FrN{\rho-\tilde{\rho}}\leq\eps$ in $\poly(1/\eps)$ time.
	\end{proposition}
	
	\begin{proposition}\label{prop:ptomo}
		Given access to copies of a single-qubit unitary $U$, one can output an estimation $\tilde{U}$ with $d_\diamond(U,\tilde{U})\leq\eps$ in $\poly(1/\eps)$ time.
	\end{proposition}
	
	The following simple lemma is particularly useful, which bounds the difference between states through a channel:
	
	\begin{lemma}\label{lemma:frob}
		Let $\Upphi$ be a quantum channel that takes $k$ qubits as the input. For every input states $\rho$ and $\rho'$, we have
		\[\FrN{\Upphi(\rho)-\Upphi(\rho')}\leq 2^{k/2}\FrN{\rho-\rho'}.\]
	\end{lemma}
	\begin{proof}
		Since quantum channels do not increase trace distance, we have
		\[\FrN{\Upphi(\rho)-\Upphi(\rho')}\leq \TrN{\Upphi(\rho)-\Upphi(\rho')} \leq \TrN{\rho-\rho'} \leq 2^{k/2}\FrN{\rho-\rho'}.
		\qedhere\]
	\end{proof}
	
	As we are frequently dealing with differences between quantum states, here we present some facts about the space of such differences. We start from the Bloch sphere presentation of a single-qubit state:
	\begin{equation}
	\rho=\frac{1}{2}(I+r_xX+r_yY+r_zZ),\qquad r_x,r_y,r_z\in\R,\quad r_x^2+r_y^2+r_z^2\leq 1
	\end{equation}
	where
	\begin{equation}
	I=\begin{bmatrix}1&0\\0&1\end{bmatrix},
	X=\begin{bmatrix}0&1\\1&0\end{bmatrix},
	Y=\begin{bmatrix}0&-i\\i&0\end{bmatrix},
	Z=\begin{bmatrix}1&0\\0&-1\end{bmatrix}.
	\end{equation}
	Then the difference $\rho-\rho'$ between two single-qubit states can be written in the Pauli basis
	\begin{equation}
	\rho-\rho'=r_xX+r_yY+r_zZ, \quad r_x^2+r_y^2+r_z^2=\frac{1}{2}\FrN{\rho-\rho'}^2\leq 1.
	\end{equation}
	Therefore, if we view $\rho-\rho'$ under the coordinate system $\sqrt{2}\cdot(r_x,r_y,r_z)$, then the set $\del_1$ of all possible single-qubit differences can be identified with a ball of radius $\sqrt{2}$ in $\R^3$. The Euclidean space $\R^3$ is equipped with the standard trace inner product $\langle A,B\rangle=\Tr[A\ct B]$, so that the norm coincides with the Frobenius norm.
	
	More generally, let $\del_k$ be the set of all possible differences between two $k$-qubit states $\rho$ and $\rho'$. The difference $\rho-\rho'$ can be written as a real linear combination over the Pauli basis
	\[\{I,X,Y,Z\}^{\otimes k}\setminus \{I^{\otimes k}\}\]
	where the identity is removed as $\Tr[\rho-\rho']=0$. Since the Pauli basis are orthonormal, we can think of $\del_k$ as a subset of the Euclidean space $\R^{4^k-1}$ (although the set is much more complicated than a ball for $k>1$). The Euclidean space is also equipped with the standard trace inner product and the Frobenius norm. As a result, a quantum channel $\Upphi$ with $k$-qubit input and $n$-qubit output induces a linear map from $\del_k$ to $\del_n$:
	\[\rho-\rho'\mapsto \Upphi(\rho)-\Upphi(\rho'),\]
	which is also a real linear map from $\R^{4^k-1}$ to $\R^{4^n-1}$.
	
	\section{Anti-concentration Bound}\label{sect:anticc}
	
	In this section we prove \Cref{thm:anticc}. For simplicity, we introduce the notion of \emph{semi-polynomials}: A function is a degree-$d$ semi-polynomial in complex variables $z_1,\ldots,z_n$, if it is a polynomial in $z_1,\ldots,z_n,\overline{z_1},\ldots,\overline{z_n}$ of degree at most $d$. Now \Cref{thm:anticc} is implied by the following more general form:
	\begin{theorem}\label{thm:anticc2}
		Let $m\leq n$, and $F:\C^{nm}\to\C$ be a degree-$d$ semi-polynomial. Suppose that $F$ takes as inputs the entries of the first $m$ columns of an $n\times n$ unitary matrix, and that the value of $F$ is always a non-negative real number over this domain. Then for $U\sim\haar(n)$,
		\[\Pr\left[F(U_{1,1},\ldots,U_{n,m})\leq\eps \E\left[F\right]\right]\leq C'(n,m,d)\cdot \eps^{C(n,m,d)}\]
		holds for every $\eps>0$, where $C(n,m,d)>0$ and $C'(n,m,d)>0$ are constants that depend only on $n,m$ and $d$.
	\end{theorem}
	
	To see that \Cref{thm:anticc2} implies \Cref{thm:anticc}, it suffices to take $m=n$ and notice that $|F|^2=F\overline{F}$ is a degree-$2d$ semi-polynomial that is always non-negative. We prove \Cref{thm:anticc2} via induction, and the proof is divided into four stages. 
	
	\subsection{\texorpdfstring{$m=1,n=1$}{m=1,n=1}}
	We start with the simplest case when $m=n=1$. In this case $F:\C\to\C$ is a single-variable degree-$d$ semi-polynomial over the unit circle $\{z:|z|=1\}$. Since $\overline{z}=1/z$ over this domain, assuming $F\neq 0$ we can write $F$ as
	\[F(z)=G(z)/z^d\]
	where $G(z)=\alpha(z-z_1)\cdots(z-z_{2d})$ is a degree $2d$ polynomial in $z$. Without loss of generality we can assume that $|\alpha|=1$, and therefore
	\begin{equation}
	F(z)=|F(z)|=|G(z)|=|z-z_1|\cdots|z-z_{2d}|.
	\end{equation}
	Then we can bound the expectation of $F$ as
	\begin{equation}
	\E[F]\leq\sup_{|z|=1}\prod|z-z_i|\leq\prod(1+|z_i|).
	\end{equation}
	
	On the other hand, if $|z-z_i|>\delta\geq0$ for some $|z|\leq 1$, it is easy to show that
	\begin{equation}
	\frac{|z-z_i|}{1+|z_i|}>\frac{\delta}{\delta+2}.
	\end{equation}
	Thus if $|z-z_i|>\delta$ holds for all $i=1,\ldots,2d$, then
	\begin{equation}
	F(z)=\prod|z-z_i|>\left(\frac{\delta}{\delta+2}\right)^{2d}\prod(1+|z_i|)\geq\left(\frac{\delta}{\delta+2}\right)^{2d}\E[F].
	\end{equation}
	That means, if we take $\eps=\left(\frac{\delta}{\delta+2}\right)^{2d}$, then $F(z)\leq\eps\E[F]$ only happens when $z$ falls into one of the $\delta$-balls around some $z_i$. Each $\delta$-ball intersect with the unit circle as an arc of angle at most $4\delta$, and thus has measure at most $2\delta/\pi$ under the Haar measure over the unit circle. Therefore we conclude that
	\begin{align}
	\Pr_{|z|=1}\left[F(z)\leq\eps\E[F]\right] &\leq \min\{4d\delta/\pi,1\} \nonumber\\
	&=\min\left\{\frac{8d}{\pi}\cdot\frac{\eps^{1/(2d)}}{1-\eps^{1/(2d)}},1\right\}
	\leq\left(\frac{8d}{\pi}+1\right)\eps^{1/(2d)}, \label{eq:a1}
	\end{align}
	and we can take $C(1,1,d)=1/(2d)$ and $C'(1,1,d)=8d/\pi+1$.
	
	\subsection{\texorpdfstring{$m=1,n=1$}{m=1,n=1}, Alternative Distribution}
	
	For the sake of later use, we also need a version where $z=u_1$ follows the distribution of the first coordinate of a Haar-random unit vector $(u_1,\ldots,u_n)\in\C^n$, $n\geq 2$. In this case $\overline{z}=1/z$ no longer holds, and we need an alternative method.
	
	For every $r\in\R$, $0\leq r\leq 1$ we define
	\[P(r)=\E_{|u_1|=r}[F(u_1)]\]
	where the expectation is over a Haar-random $z\in\C$ with $|z|=r$. Notice that a monomial $u_1^k\overline{u_1}^{\ell}$ in $F$ has expectation $0$ unless $k=\ell$, and when $k=\ell$ we have $u_1^k\overline{u_1}^{\ell}=r^{2k}$. That means $P(r)$ is a degree-$d$ polynomial in $r^2$.
	
	Notice that $r^2$ follows the Beta distribution $\mathrm{Beta}(1,n-1)$, with the density function
	\[f(r;1,n-1)=(n-1)(1-r)^{n-2},\]
	and $\E[P(r)]$ under this distribution coincides with $\E[F]$.
	
	With the analysis in the previous section which also works on $P(r)$, we can show that if we take $\eps=\left(\frac{\delta}{\delta+2}\right)^d$, then $P(r)\leq\eps\E[P(r)]=\eps\E[F]$ only happens when $r^2$ falls into one of the $\delta$-balls around $d$ complex roots of $P$, which are intervals of length at most $2\delta$ on the real line. Since the density function of $r^2$ has a maximum of $n-1$, we have
	\begin{equation}
	\Pr[P(r)\leq\eps\E[F]]\leq\min\{2d(n-1)\delta,1\}\leq 4dn\eps^{1/d}.
	\end{equation}
	
	On the other hand, applying \cref{eq:a1} from the previous section on $F(rz)$ for every fixed $r$ directly provides
	\begin{equation}
	\Pr_{|z|=r}[F(z)\leq\eps P(r)]\leq \left(\frac{8d}{\pi}+1\right)\eps^{1/(2d)}.
	\end{equation}
	Thus by a union bound we have
	\begin{align}
	\Pr[F(u_1)\leq\eps\E[F]]
	&\leq \Pr\left[P(r)\leq\sqrt{\eps}\E[F]\right]+\Pr\left[F(z)\leq\sqrt{\eps} P(|z|)\right] \nonumber\\
	&\leq 4dn\eps^{1/(2d)}+\left(\frac{8d}{\pi}+1\right)\eps^{1/(4d)} \nonumber\\
	&\leq 4d(n+1)\eps^{1/(4d)}. \label{eq:a2}
	\end{align}
	We note that not only this result will be used in the next stage, the technique itself will also be reapplied multiple times in the later proofs.
	
	\subsection{\texorpdfstring{$m=1,n>1$}{m=1,n>1}}
	In this stage we consider $m=1$ with general $n$, and thus the inputs to $F$ is a Haar-random unit vector $(u_1,\ldots,u_n)\in\C^n$. The strategy is to use induction on $n$, and show that with high probability over the choice of $u_1$,
	\[P(u_1)=\E_{u_2,\ldots,u_n}[F(u_1,\ldots,u_n)]\]
	is not too small conditioned on the fixed $u_1$. To handle this, we need the following lemma.
	
	\begin{lemma}
		If $F:\C^n\to\C$ is a degree-$d$ semi-polynomial, and $(u_1,\ldots,u_n)\in\C^n$ is a Haar-random unit vector, then
		$P(u_1)=\E\nolimits_{u_2,\ldots,u_n}[F(u_1,\ldots,u_n)]$
		is a degree-$d$ semi-polynomial on $u_1$.
	\end{lemma}
	\begin{proof}
		Consider each monomial in $F$, and let $G$ be the part of monomial over $u_2,\ldots,u_n$ and their conjugates. Since $(u_2,\ldots,u_n)=r\cdot u'$, where $r=(1-|u_1|^2)^{1/2}$ and $u'$ follows the Haar measure over the unit sphere in $\C^{n-1}$, we have:
		\begin{itemize}
			\item If $G$ has an odd degree then $\E[G]=0$, by the symmetry $u'\to-u'$;
			\item And if $G$ has an even degree $2\ell\leq d$ then
			\begin{equation}
			\E[G]=r^{2\ell}\E[G(u')]=(1-u_1\overline{u_1})^\ell\E[G(u')],
			\end{equation}
			where $\E[G(u')]$ is a constant irrelevant to the choice of $u_1$.
		\end{itemize}
		Either way $\E[G]$ is a degree-$d$ semi-polynomial in $u_1$, and thus so is $P(u_1)$.
	\end{proof}
	
	Since $P$ is an expectation over $F$, it is also always non-negative and has the same expectation as $\E[F]$. Thus \cref{eq:a2} from the previous section gives
	\begin{equation}
	\Pr\left[P(u_1)\leq\eps\E[F]\right] \leq 4d(n+1)\eps^{1/(4d)}.
	\end{equation}
	
	Now we fix some $u_1$, and consider the degree-$d$ semi-polynomial
	\[F_{u_1}(u_2,\ldots,u_n)=F(u_1,u_2,\ldots,u_n)\]
	which is always non-negative. Since $(u_2,\ldots,u_n)=ru'$ for some $r\in\R$, and $u'$ follows the Haar measure over the unit sphere in $\C^{n-1}$, we can apply the induction hypothesis for $n-1$ on $F_{u_1}(ru')$ to get
	\begin{equation}
	\Pr\left[F_{u_1}(ru')\leq\eps\E\nolimits_{u'}[F_{u_1}(ru')]\right]\leq C'(n-1,1,d)\cdot \eps^{C(n-1,1,d)}.
	\end{equation}
	Therefore we conclude that, for every $p\in(0,1)$,
	\begin{align}
	&\Pr\left[F(u_1,\ldots,u_n)\leq\eps\E[F]\right] \nonumber\\
	\leq\ & \Pr\left[P(u_1)\leq\eps^p\E[F]\right]+\Pr\left[F(u_1,\ldots,u_n)\leq\eps^{1-p} P(u_1)\right] \nonumber\\
	=\ & \Pr\left[P(u_1)\leq\eps^p\E[F]\right]+\Pr\left[F_{u_1}(u_2,\ldots,u_n)\leq\eps^{1-p}\E[F_{u_1}]\right] \nonumber\\
	\leq\ & 4d(n+1)\eps^{p/(4d)}+C'(n-1,1,d)\cdot \eps^{(1-p)C(n-1,1,d)}. \label{eq:a3}
	\end{align}
	We can take 
	\[C(n,1,d)=\max_p\min\left\{\frac{p}{4d}, (1-p)C(n-1,1,d)\right\}=\frac{1}{4d+C(n-1,1,d)^{-1}}=\frac{1}{4nd},\]
	and $C'(n,1,d)=4d(n+1)+C'(n-1,1,d)=O(n^2d)$.
	
	\subsection{\texorpdfstring{$m>1,n>1$}{m>1,n>1}}
	Now we handle the general case when the inputs to the semi-polynomial consist of $m$ columns of a Haar random unitary. The proof is similar to the last stage, using the fact that the input distribution can be viewed as a unitary-invariant distribution over $m$ orthonormal vectors in $\C^n$. In particular, let the vectors be $v_1,\ldots,v_m$, and we consider the function
	\[P(v_1,\ldots,v_{m-1})=\E_{v_m}[F(v_1,\ldots,v_m)].\]
	Here $v_m$ is a Haar-random vector over the unit sphere in the $(n-m)$-dimensional orthogonal subspace of $\mathrm{span}(v_1,\ldots,v_m)$. We would like to show that $P$ is a semi-polynomial in order to use induction, and we first show it with $v_m$ replaced with a Gaussian distribution.
	
	\begin{lemma}\label{lemma:gaussian}
		Let $v_1,\ldots,v_m$ be a set of $m<n$ orthonormal vectors in $\C^n$, and let $g=(g_1,\ldots,g_n)\in\C^n$ distributed as a standard $(n-m)$-dimensional complex Gaussian in the orthogonal subspace of $\mathrm{span}(v_1,\ldots,v_m)$.
		
		Let $G:\C^n\to\C$ be a degree-$d$ semi-polynomial in $(g_1,\ldots,g_n)$. Then
		\[Q(v_1,\ldots,v_m)=\E[G(g_1,\ldots,g_n)]\]
		is a degree-$d$ semi-polynomial in the entries of $v_1,\ldots,v_m$.
	\end{lemma}
	\begin{proof}
		Notice that $g=(g_1,\ldots,g_n)$ can be obtained by taking a standard $n$-dimensional complex Gaussian $g'\in\C^n$, and apply the Gram-Schmidt orthogonalization:
		\begin{equation}
		g=g'-\sum_{i=1}^m v_iv_i\ct g'.
		\end{equation}
		Therefore the covariance matrix of $g$ is $\left(I_n-\sum v_iv_i\ct\right)^2=I_n-\sum v_iv_i\ct$, where each entry is a degree-$2$ semi-polynomial in $v_1,\ldots,v_m$. By the complex Wick's theorem \cite{Fassino19}, the expectation of a monomial in $G(g)$ 
		\[\E\left[g_1^{\alpha_1}\overline{g_1}^{\beta_1}\cdots g_n^{\alpha_n}\overline{g_n}^{\beta_n}\right]\]
		is non-zero only when $\sum\alpha_i=\sum\beta_i$, in which case it is
		a polynomial function over the entries of the covariance matrix of $g$ with degree $\sum\alpha_i$. Therefore $Q=\E[G]$ is a degree-$d$ semi-polynomial in the entries of $v_1,\ldots,v_m$. 
	\end{proof}
	\begin{corollary}
		$P(v_1,\ldots,v_{m-1})=\E\nolimits_{v_m}[F(v_1,\ldots,v_m)]$ is a degree-$d$ semi-polynomial in the entries of $v_1,\ldots,v_{m-1}$.
	\end{corollary}
	\begin{proof}
		Notice that $v_m$ is equidistributed as $g/\|g\|_2$, where $g=(g_1,\ldots,g_n)$ is the Gaussian in \Cref{lemma:gaussian}. In other words, $g=r\cdot v_m$ where $r\in\R$ follows a fixed $\chi$ distribution independent of $v_m$.
		
		Now consider each monomial in $F$, and let $G$ be the part of monomial over entries of $v_m$, then it suffices to show that $\E[G(v_m)]$ is a degree-$d$ semi-polynomial in $v_1,\ldots,v_{m-1}$. \Cref{lemma:gaussian} already showed this for $\E[G(g)]$, and since $G$ is a monomial of degree $\ell\leq d$, we have
		\begin{equation}
		\E[G(g)]=\E[G(r\cdot v_m)]=\E[r^\ell]\E[G(v_m)],
		\end{equation}
		where $\E[r^\ell]$ is a non-zero constant irrelevant to the choice of $v_1,\ldots,v_{m-1}$. Therefore $\E[G(v_m)]$ is also a degree-$d$ semi-polynomial in $v_1,\ldots,v_m$.
	\end{proof}
	
	Since $P$ is an expectation over $F$, it is also always non-negative and has the same expectation as $\E[F]$. Thus the induction hypothesis gives
	\begin{equation}
	\Pr\left[P(v_1,\ldots,v_{m-1})\leq\eps\E[F]\right]\leq C'(n,m-1,d)\cdot \eps^{C(n,m-1,d)}.
	\end{equation}
	
	Now we fix some $v_1,\ldots,v_{m-1}$ and consider the degree-$d$ semi-polynomial
	\[F_{v_1,\ldots,v_{m-1}}(v_m)=F(v_1,\ldots,v_{m-1},v_m),\]
	which is always non-negative. Since there exists a linear map $A:\C^{n-m+1}\to\C^n$ such that $v_m=Au$, where $u$ is a Haar-random unit vector in $\C^{n-m+1}$, we can apply \cref{eq:a3} from the previous stage for $m=1$ on $F_{v_1,\ldots,v_{m-1}}(Au)$ to get
	\begin{equation}
	\Pr\left[F_{v_1,\ldots,v_{m-1}}(Au)\leq\eps\E\nolimits_u\left[F_{v_1,\ldots,v_{m-1}}(Au)\right]\right]\leq C'(n,1,d)\cdot \eps^{C(n,1,d)}.
	\end{equation}
	Therefore we conclude that, for every $p\in(0,1)$,
	\begin{align}
	&\Pr\left[F(v_1,\ldots,v_m)\leq\eps\E[F]\right] \nonumber\\
	\leq\ & \Pr\left[P(v_1,\ldots,v_{m-1})\leq\eps^p\E[F]\right]
	+\Pr\left[F(v_1,\ldots,v_m)\leq\eps^{1-p}P(v_1,\ldots,v_{m-1})\right] \nonumber\\
	=\ & \Pr\left[P(v_1,\ldots,v_{m-1})\leq\eps^p\E[F]\right]
	+\Pr\left[F_{v_1,\ldots,v_{m-1}}(v_m)\leq\eps^{1-p}\E[F_{v_1,\ldots,v_{m-1}}]\right] \nonumber\\
	\leq\ & C'(n,m-1,d)\cdot \eps^{pC(n,m-1,d)}+C'(n,1,d)\cdot \eps^{(1-p)C(n,1,d)},
	\end{align}
	Similar to the previous stage, we can take 
	\[C(n,m,d)=\frac{1}{C(n,m-1,d)^{-1}+ C(n,1,d)^{-1}}=\frac{1}{4nmd},\]
	and $C'(n,m,d)=C'(n,m-1,d)+C'(n,1,d)=O(n^2md)$. This completes the proof of \Cref{thm:anticc2}.
	
	\section{Mixing Bound for Random Circuits}\label{sect:mixing}
	
	In this section we prove \Cref{thm:mixing}. Since the output qubit $\phic(\rho)$ is in the lightcone of the input qubit $\rho$, there exists gates $G_1,\ldots,G_D$ in the circuit $\mc C$ that connect the input and output qubits. That is, there are qubits $\rho_0,\rho_1,\ldots,\rho_D$ with $\rho_0=\rho$ and $\rho_D=\phic(\rho)$, such that the gate $G_i$ has $\rho_{i-1}$ as an input and $\rho_i$ as an output. Each gate $G_i$ is independently drawn from the Haar measure $\haar(2^k)$.
	
	Note that even when gate $G_i$ is given, we cannot directly claim any relationship between $\rho_{i-1}$ and $\rho_i$, as the other input qubits to $G_i$ are correlated and possibly entangled with $\rho_{i-1}$. To handle this, we let $\tau_{i-1}$ be the composite state of the $k$ input qubits to $G_i$, and by \Cref{lemma:frob} we have
	\begin{equation}\label{eq:m1}
	\FrN{\tau_{i-1}-\tau'_{i-1}}\geq 2^{-k/2}\FrN{\rho_{i-1}-\rho'_{i-1}},
	\end{equation}
	where $\rho'_i$ and $\tau'_i$ are the corresponding states when the input state is $\rho'_0=\rho'$. Therefore we only need to bridge the remaining gaps by proving
	\begin{equation}
	\FrN{\rho_i-\rho'_i}\geq \lambda_i\FrN{\tau_{i-1}-\tau'_{i-1}}
	\end{equation}
	where $\lambda_i$ is some function of $G_i$, and then bound the distribution of $\lambda_i$ by applying the anti-concentration bound from \Cref{thm:anticc}.
	
	At a first glance, this looks impossible as $\tau_{i-1}$ is a $k$-qubit state while $\rho_i$ is single-qubit, which means that as long as $k>1$, whatever $G_i$ is, there will be input states $\tau_{i-1}\neq\tau'_{i-1}$ to $G_i$ with the output qubits $\rho_i=\rho'_i$. The key observation is that the states $\tau_i$ cannot be arbitrary $k$-qubit states: Since all the input qubits to the circuit $\mc C$ are fixed except $\rho$, when the circuit $\mc C$ is given, each $\tau_i$ is the result of $\rho$ through a fixed quantum channel. As the difference $\rho-\rho'$ ranges within the 3-dimensional Euclidean space $\Updelta_1$, the difference $\tau_{i-1}-\tau'_{i-1}$ could also only range within a 3-dimensional subspace of $\Updelta_k$.
	
	This allows us to define and bound the distribution of $\lambda_i$ that uniformly holds for every such subspace as follows:
	
	\begin{lemma}\label{lemma:mixing}
		For every $k\in\mathbb{N}$, there exists a distribution $\Uplambda_k$ over $[0,\infty)$ that satisfies: 
		\begin{itemize}
			\item There exists $t>0$ such that $\E[\Uplambda_k^{-t}]<\infty$, where we use $\E[\Uplambda_k^{-t}]$ as shorthand for $\E\nolimits_{x\sim\Uplambda_k}[x^{-t}]$.
			\item For every fixing of the input qubits other than $\rho$, and layers $1,\ldots,i-1$ of the circuit $\mc C$, there exists a function $\lambda_i\colon\ugrp(2^k)\to[0,\infty)$ at layer $i$ such that
			\[\FrN{\rho_i-\rho'_i}\geq \lambda_i(G_i)\FrN{\tau_{i-1}-\tau'_{i-1}}\]
			holds for all input states $\rho$ and $\rho'$, and $\lambda_i(G_i)\sim\Uplambda_k$ when $G_i\sim\haar(2^k)$.
		\end{itemize}
	\end{lemma}
	
	We will prove \Cref{lemma:mixing} in \Cref{sect:proof_mixing}. For now let us show how \Cref{lemma:mixing} would imply \Cref{thm:mixing}.
	
	\begin{proof}[Proof of \Cref{thm:mixing}]
		From \Cref{lemma:mixing} and \cref{eq:m1} we get
		\begin{equation}
		\FrN{\rho_i-\rho'_i}\geq \lambda_i(G_i)\FrN{\tau_{i-1}-\tau'_{i-1}}
		\geq 2^{-k/2}\lambda_i(G_i)\FrN{\rho_{i-1}-\rho'_{i-1}}
		\end{equation}
		for every $i=1,\ldots,D$. Here each function $\lambda_i(G_i)$ depends on the previous layers, but as random variables $\lambda_i=\lambda_i(G_i)$ they are independent, since $\lambda_i$ follows the same distribution $\Uplambda_k$ no matter how $\lambda_1,\ldots,\lambda_{i-1}$ are fixed.
		
		Since $\FrN{\rho_0-\rho'_0}=\FrN{\rho-\rho'}$ and $\FrN{\rho_D-\rho'_D}=\FrN{\phic(\rho)-\phic(\rho')}$, we have
		\begin{equation}
		\FrN{\phic(\rho)-\phic(\rho')}\geq 2^{-kD/2}\lambda_1\cdots\lambda_D\FrN{\rho-\rho'}.
		\end{equation}
		To bound the product of $\lambda_i$ we use Markov's inequality, which states that for every $\alpha>0$,
		\begin{align} 
		\Pr[\lambda_1\cdots\lambda_D\leq\alpha]
		&= \Pr[(\lambda_1\cdots\lambda_D)^{-t}\geq \alpha^{-t}] \nonumber\\
		&\leq \alpha^t\E[\lambda_1^{-t}\cdots\lambda_D^{-t}] = \alpha^t\E[\Uplambda_k^{-t}]^D
		\end{align}
		where we take $t>0$ to be the constant in \Cref{lemma:mixing}. Take $\alpha$ such that $\gamma=\alpha^t\E[\Uplambda_k^{-t}]^D$, we conclude that with probability at least $1-\gamma$,
		\begin{align*}
		\FrN{\phic(\rho)-\phic(\rho')} &\geq 2^{-kD/2}\alpha\FrN{\rho-\rho'} \\
		&= 2^{-kD/2}\gamma^{1/t} \E[\Uplambda_k^{-t}]^{-D/t}\FrN{\rho-\rho'} \\
		&= (2^{-D}\gamma)^{O_k(1)}\FrN{\rho-\rho'}. \qedhere
		\end{align*}
	\end{proof}
	
	\subsection{Proof of \texorpdfstring{\Cref{lemma:mixing}}{Lemma 5.1}}\label{sect:proof_mixing}
	
	The basic idea of the proof is to lower bound the ratio $\FrN{\rho_i-\rho'_i}/\FrN{\tau_{i-1}-\tau'_{i-1}}$ with a polynomial function on the entries of $G_i$ and apply \Cref{thm:anticc}. 
	
	The quantum channel defined by $G_i$ that maps $\tau_{i-1}$ to $\rho_i$ induces the linear map 
	\[M_i:\tau_{i-1}-\tau'_{i-1}\mapsto\rho_i-\rho_i'.\]
	The domain of $M_i$ is a 3-dimensional subspace of $\del_k$ which we denote as $\mc S_i$, while the range of $M_i$ is $\del_1$. Notice that the map $M_i$ is completely determined by the domain $\mc S_i$ and the gate $G_i$, while $\mc S_i$ depends only on the fixed inputs to the circuit $\mc C$ and the gates of $\mc C$ in layers $1,\ldots,i-1$.
	
	Since both the domain and the range are Euclidean spaces, the absolute determinant $\left|\det M_i\right|$ is independent of the choices of bases when writing $M_i$ as matrix in $\R^{3\times 3}$. We will show that $\left|\det M_i\right|$ is basically the lower bound that we seek for, via the following propositions.
	
	\begin{proposition}\label{prop:bound}
		It always holds that
		\[\FrN{\rho_i-\rho'_i}\geq 2^{-k}\left|\det M_i\right|\cdot \FrN{\tau_{i-1}-\tau'_{i-1}}.\]
	\end{proposition}
	\begin{proof}
		Let $\sigma_1\leq\sigma_2\leq\sigma_3$ be the singular values of $M_i$, then $\FrN{\rho_i-\rho'_i}\geq \sigma_1\FrN{\tau_{i-1}-\tau'_{i-1}}$ as the norms in both spaces coincide with the Frobenius norm. On the other hand, we know that $\sigma_2\leq\sigma_3\leq 2^{k/2}$ by \Cref{lemma:frob}. Therefore, $\left|\det M_i\right|=\sigma_1\sigma_2\sigma_3\leq 2^k \sigma_1$ and thus the claim holds.
	\end{proof}
	
	\begin{proposition}\label{prop:degree}
		For each fixed domain $\mc S_i$, $\det M_i$ is a degree-6 semi-polynomial in the entries of $G_i$.
	\end{proposition}
	\begin{proof}
		After fixing the orthonormal bases $\{\xi_1,\xi_2,\xi_3\}$ for $\mc S_i$ and $\{\sigma_1,\sigma_2,\sigma_3\}$ for $\del_1$, the $(\ell,r)$-th entry in the matrix representation of $M_i$ is
		\begin{equation}\label{eq:deg}
		\Tr[\sigma_r\cdot M_i(\xi_\ell)]=\Tr[(I_A\otimes \sigma_r)\cdot G_i\xi_\ell G_i\ct]
		\end{equation}
		where $A$ is the system that consists of the output qubits of $G_i$ other than $\rho_i$. This is a quadratic form in the entries of $G_i$ and thus a degree-2 semi-polynomial, and therefore $\det M_i$ is a degree-6 semi-polynomial.
	\end{proof}
	
	\begin{proposition}\label{prop:minimum}
		There exists a constant $\mu_k>0$ such that for every possible domain $\mc S_i$,
		\[\E_{G_i\sim\haar(2^k)}\left[\left|\det M_i\right|^2\right]\geq\mu_k.\]
	\end{proposition}
	\begin{proof}
		When $\mc S_i$ is fixed, $\left|\det M_i\right|^2$ is a continuous function of $G_i\in\ugrp(2^k)$. For at least one $G_i$, which is a permutation over the $k$ qubits that swaps $\rho_{i-1}$ to $\rho_i$, we have $\left|\det M_i\right|>0$. This implies that $\E\left[\left|\det M_i\right|^2\right]>0$ always holds.
		
		Now we think of $\E\left[\left|\det M_i\right|^2\right]$ as a continuous function of $\mc S_i$, while the set of all possible $\mc S_i$ is a closed subset of the Grassmannian $\mathbf{Gr}_3(\R^{4^k-1})$ and thus is compact. That means the function admits a global minimum $\mu_k$, which depends only on $k$, and $\mu_k>0$.
	\end{proof}
	
	Now we are ready to prove \Cref{lemma:mixing}.
	\begin{proof}[Proof of \Cref{lemma:mixing}]
		Applying \Cref{thm:anticc} with \Cref{prop:degree,prop:minimum} gives
		\begin{align}
		\Pr\left[2^{-k}\left|\det M_i\right|\leq\eps\right]
		&\leq \Pr\left[\left|\det M_i\right|^2\leq 2^{2k}\mu_k^{-1}\eps^2\E\left[\left|\det M_i\right|^2\right]\right] \nonumber\\
		&\leq C'(2^k,6)\cdot (2^{2k}\mu_k^{-1}\eps^2)^{C(2^k,6)} \label{eq:m2}
		\end{align}
		which holds for every domain $\mc S_i$ and every $\eps\geq 0$. We define $\Uplambda_k$ as the distribution over $[0,\infty)$ with the following cumulative function:
		\[\Pr[\Uplambda_k\leq x]=\min\left\{C'(2^k,6)\cdot (2^{2k}\mu_k^{-1} x^2)^{C(2^k,6)},1\right\}.\]
		Take $t=C(2^k,6)>0$ and let $C=C'(2^k,6)\cdot (2^{2k}\mu_k^{-1})^{C(2^k,6)}>0$, then we have
		\begin{align}
		\E[\Uplambda_k^{-t}]&=\int_0^\infty x^{-t}\dd\Pr[\Uplambda_k\leq x] \nonumber\\
		&= \int_0^{C^{-1/(2t)}} x^{-t}\dd(Cx^{2t}) \nonumber\\
		&=\int_0^{C^{-1/(2t)}} 2tC\cdot x^{t-1}\dd x = 2tC^{1/2}<\infty.
		\end{align}
		
		Now suppose the input qubits other than $\rho$ are fixed, and the layers $1,\ldots,i-1$ of the circuit $\mc C$ is given. This fixes the domain $\mc S_i$ for $M_i$, and provides the cumulative function
		\[P(x)=\Pr\left[2^{-k}\left|\det M_i\right|\leq x\right].\]
        Notice that $P$ is continuous and non-decreasing, and by defining its inverse $P^{-1}(y)=\sup_{x} P(x)\leq y$ on $[0,1]$, we can show that for every $y\in[0,1]$,
        \[\Pr\left[P\big(2^{-k}\left|\det M_i\right|\big)\leq y\right]=\Pr\left[2^{-k}\left|\det M_i\right|\leq P^{-1}(y)\right]=P(P^{-1}(y))=y.\]
        
		We then define the function $\lambda_i:\ugrp(2^k)\to[0,\infty)$ as follows: For each $G_i\in\ugrp(2^k)$, let $\lambda_i(G_i)$ be the smallest $\lambda\geq 0$ such that
		\[P\big(2^{-k}\left|\det M_i\right|\big)=\Pr[\Uplambda_k\leq \lambda].\]
		When $G_i\sim\haar(2^k)$, we have $\lambda_i(G_i)\sim\Uplambda_k$ since
		\begin{equation}
		\Pr[\lambda_i(G_i)\leq x]=\Pr\left[P\big(2^{-k}\left|\det M_i\right|\big)\leq \Pr[\Uplambda_k\leq x]\right] =\Pr[\Uplambda_k\leq x].
		\end{equation}
		We also have $2^{-k}\left|\det M_i\right|\geq \lambda_i(G_i)$ since $P(x)\leq \Pr[\Uplambda_k\leq x]$ holds for all $x\geq 0$. Combined with \Cref{prop:bound} we get
		\[\FrN{\rho_i-\rho'_i}\geq 2^{-k}\left|\det M_i\right|\cdot \FrN{\tau_{i-1}-\tau'_{i-1}}
		\geq \lambda_i(G_i)\FrN{\tau_{i-1}-\tau'_{i-1}}. \qedhere\]
	\end{proof}
	
	\section{Applications}\label{sect:app}
	
	\subsection{Depth Lower Bound for Approximate Designs}\label{sect:depthlb}
	
	In this section we prove \Cref{thm:depthlb}. We first recall the definition of an approximate unitary design.
	
	\begin{definition}
		For a distribution $\mc D$ over $\ugrp(n)$ and $t\in\mathbb{N}_+$, define the moment superoperator as the following channel:
		\[\Upphi^{(t)}_{\mc D}:\rho\mapsto\int_{\ugrp(n)}U^{\otimes t}\rho (U\ct)^{\otimes t}\dd \mc D(U).\]
		The distribution $\mc D$ is an $\eps$-approximate unitary $t$-design if 
		\[\DmN{\Upphi^{(t)}_{\mc D} - \Upphi^{(t)}_{\haar(n)}}\leq\eps.\]
	\end{definition}
	
	Note that there are several other definitions of the approximate design (see e.g. \cite{Brandao12}), and we choose the weaker one so that our \Cref{thm:depthlb} is still compatible with the stronger definitions of designs.
	
	\begin{proof}[Proof of \Cref{thm:depthlb}]
		Without loss of generality, let us assume that the first input qubit $\rho$ and the first output qubit $\pi$ are depth $D$ apart. We fix all other input qubits to be maximally mixed, so that when $\rho$ is also maximally mixed, the entire output is maximally mixed regardless of the circuit $\mc C$, and thus $\pi=I/2$. On the other hand, when $\rho=\ket{0}\bra{0}$, we know the following about the output qubit $\pi$ via \Cref{thm:mixing} that with probability $1-2^{-D}$ over the circuit $\mc C$,
		\begin{equation}\label{eq:d1}
		\FrN{\pi-I/2}\geq \FrN{\ket{0}\bra{0}-I/2}\cdot (2^{-2D})^{c_k} = \frac{1}{\sqrt{2}}\cdot 2^{-2Dc_k}.
		\end{equation}
		We can expand the left hand side of \cref{eq:d1} as
		\begin{equation}
		\FrN{\pi-I/2}^2=\Tr[(\pi-I/2)^2]=\Tr[\pi^2]-1/2.
		\end{equation}
		Since $\Tr[\pi^2]\leq 1$ always holds, we get
		\begin{equation}
		\E\nolimits_{\mc C}[\Tr[\pi^2]]\geq \frac{1}{2} + \frac{1}{2}\cdot 2^{-4Dc_k} + \frac{1}{2}\cdot 2^{-D}.
		\end{equation}
		
		Now imagine feeding two copies of the input state $\ket{0}\bra{0}\otimes (I/2)^{\otimes (n-1)}$ to the superoperators $\Upphi^{(2)}_{\mc C}$ and $\Upphi^{(2)}_{\haar(n)}$, and apply a swap test on the first output qubits in the two copies. The output probability is determined by $\Tr[\pi^2]$. We know already from \cite{Emerson05} that when going through an $n$-qubit Haar random unitary, we have $\E\nolimits_{\mc U(n)}[\Tr[\pi^2]]=1/2$. Therefore, we conclude that the difference $2^{-4Dc_k}+2^{-D}\leq O(\eps)$ which means that $D\geq \Omega_k(\log \eps^{-1})$.
	\end{proof}
	
	\subsection{Depth Test}\label{sect:depthtest}
	
	In this section we prove \Cref{thm:depthtest}, where we learn the exact depth of a brickwork random circuit $\mc C$.
	
	\begin{algorithm}[ht]
		\DontPrintSemicolon
		Arbitrarily fix all input qubits except the first one $\rho$. \\
		\For {$D=0,1,\ldots$}{
			Let $\eps=(2^{-2D}\gamma)^{c_2}/4$; \\
			Apply $\mc C$ with $\rho=\ket{0}\bra{0}$ and let $\pi$ be the $(D+2)$-th output qubit; \\
			Apply $\mc C$ with $\rho=\ket{1}\bra{1}$ and let $\pi'$ be the $(D+2)$-th output qubit; \\
			Estimate $\pi$ and $\pi'$ up to $\eps$ error by state tomography; \\
			\lIf {$\FrN{\pi-\pi'} \leq 2\eps$}{\Return {$D$}.}
		}
		\caption{Algorithm for depth testing.}\label{alg:depth}
	\end{algorithm}
	
	The processed is described in \Cref{alg:depth}. Here $c_2$ is the constant $c_k$ in \Cref{thm:mixing} with $k=2$ for brickwork circuits, and $\gamma>0$ is the target error probability. Notice that in a brickwork circuit of depth $D$, the $(D+2)$-th output qubit lies outside the lightcone of the first input qubit. Therefore, when the algorithm iterates to the correct depth $D$, we have $\pi=\pi'$ and $D$ must be returned even when both $\pi$ and $\pi'$ are estimated with $\eps$ error.
	
	On the other hand, when $D$ is smaller than the actual depth, the $(D+2)$-th output qubit lies inside the lightcone of the first input qubit. By \Cref{thm:mixing}, with probability $1-2^{-D}\gamma$ over $\mc C$ we have
	\begin{equation}
	\FrN{\pi-\pi'}\geq \FrN{\ket{0}\bra{0}-\ket{1}\bra{1}}\cdot (2^{-2D}\gamma)^{c_2}>4\eps.
	\end{equation}
	This means when both $\pi$ and $\pi'$ are estimated with $\eps$ error, we still have $\FrN{\pi-\pi'}>2\eps$. By a union bound over $D$, with probability $1-\gamma$ all those $D$ smaller than the actual depth will be skipped, and thus the outputted depth is correct.
	
	Note that the efficiency of the algorithm depends on the single-qubit tomography process, which by \Cref{prop:tomo} is $\poly(1/\eps)$. As a conclusion, we obtain the following more general statement which implies \Cref{thm:depthtest}:
	
	\begin{theorem}\label{thm:depthtest2}
		Let $\mc C$ be a brickwork random quantum circuit of an unknown depth $D$ , where each gate is independently Haar random.. Given oracle access to $\mc C$, for any $\gamma\in(0,1)$, \Cref{alg:depth} outputs $D$ with probability at least $1-\gamma$ in time $\poly(2^D,\gamma^{-1})$.
	\end{theorem}
	
	\begin{remark}
		The only property of the brickwork architecture we used here is that the set of qubits within the lightcone of an input qubit is strictly expanding when the depth grows, which allows us to distinguish between different depths. Therefore the algorithm can be easily modified, with the same efficiency, to work with higher dimensional brickwork circuits and other architectures.
	\end{remark}
	
	\subsection{Learning Brickwork Random Circuits}\label{sect:learn}
	
	\subsubsection{Learning the First Gate}
	
	Here we prove \Cref{thm:learngate}, where we learn the gate $G_{1,1}$, namely the gate in the first layer acting on the first and second qubit, in a brickwork circuit of depth $D=O(\log n)$ with Haar random gates. The same arguments also work for other gates in the first layer.
	
	To learn $G_{1,1}$, we try to uncompute $G_{1,1}$ by first apply some two-qubit unitary $G\ct\in\ugrp(4)$ and then apply the circuit $\C$. We distinguish whether $G$ is close to $G_{1,1}$ or not using the similar idea as in \Cref{sect:depthtest}. Specifically, if $G=G_{1,1}$ so that $G\ct$ perfectly uncomputes the gate $G_{1,1}$, then the $(D+1)$-th output qubit will lie outside the lightcone of the first input qubit. Note that this is also true when $G\ct$ cancels $G_{1,1}$ into unentangled single-qubit gates, that is when there exists $U_1,U_2\in\ugrp(2)$ that $G_{1,1}G\ct=U_1\otimes U_2$. In contrast, when $G\ct$ does not cancel $G_{1,1}$ into unentangled single-qubit gates, the first two qubits will be entangled after the first layer of gates, and thus the $(D+1)$-th output qubit will be affected when the first input qubit changes.
	
	To put the above intuition more formally, we first define the distance between two-qubit gates when taking the quotient over unentangled single-qubit gates:
	
	\begin{definition}
		For $G,G'\in\ugrp(4)$, we define the distance
		\[d_\otimes(G,G')=\min_{U_1,U_2\in\ugrp(2)} d_\diamond(G,(U_1\otimes U_2)\cdot G').\]
		We note that $d_\otimes$, like the diamond distance $d_\diamond$, is a pseudometric on $\ugrp(4)$. That is, a metric except that two distinct unitaries could have distance zero. However, $d_\otimes(G,G')=0$ if and only if $G'\cdot G^{-1}=U_1\otimes U_2$ for some $U_1,U_2\in\ugrp(2)$.
	\end{definition}
	
	The learning algorithm is described in \Cref{alg:learngate}, with $\delta,\gamma>0$ being the target error rate. Notice that within each loop of $\sigma_2$, the second input qubit $\rho_2$ is fixed, while the first one changes with a difference $\sigma_1$ that goes through the Pauli basis. The input qubits $\rho_1$ and $\rho_2$ are Pauli eigenstates except when $\sigma_2=\rho_2=I/2$, which can be obtained by randomly choosing $\ket{0}$ or $\ket{1}$.
	
	\begin{algorithm}[ht]
		\DontPrintSemicolon
		Arbitrarily fix all input qubits except the first and second ones $\rho_1,\rho_2$. \\
		Let $\eps = 10^{-4}\delta^2(2^{-D+1}\gamma)^{c_2}$; \\
		\ForEach {$G$ from an $\eps$-net of $\ugrp(4)$ under distance $d_\otimes$}{
			\For {$\sigma_1\in\{X,Y,Z\}$ and $\sigma_2\in\{I,X,Y,Z\}$}{
				Let $\rho_1 = (I+\sigma_1)/2, \rho_2 = (I+\sigma_2)/2$ and apply $G\ct$ on the first two qubits, \\
				\Indp then apply $\mc C$ and let $\pi$ be the $(D+1)$-th output qubit; \\
				\Indm Let $\rho_1 = (I-\sigma_1)/2, \rho_2 = (I+\sigma_2)/2$ and apply $G\ct$ on the first two qubits, \\
				\Indp then apply $\mc C$ and let $\pi'$ be the $(D+1)$-th output qubit; \\
				\Indm Estimate $\pi$ and $\pi'$ up to $\eps$ error by state tomography; \\
				\lIf {$\FrN{\pi-\pi'} \geq 5\eps$}{reject $G$.}
			}
			\Return {$G$ if not rejected.}
		}
		\caption{Algorithm for learning the gate $G_{1,1}$}\label{alg:learngate}
	\end{algorithm}

    To prove the correctness of the algorithm, we need the following lemmas that connects the distance $d_\otimes(U,I\otimes I)$ with the behavior of $U$ over the Pauli basis.
	
	\begin{lemma}\label{lemma:complete}
		For $U\in\ugrp(4)$, if $d_\otimes(U,I\otimes I)\leq\delta$ then for every $\sigma_1\in\{X,Y,Z\}$ and $\sigma_2\in\{I,X,Y,Z\}$, 
		\[\FrN{\Tr_A\left[U(\sigma_1\otimes\sigma_2)U\ct\right]}\leq 2\delta,\]
		where $\Tr_A$ is the partial trace that traces out the first qubit.
	\end{lemma}
	\begin{proof}
		By definition, $d_\otimes(U,I\otimes I)\leq\delta$ means that there exists $U_1,U_2\in\ugrp(2)$ that $d_\diamond(U,U_1\otimes U_2)\leq\delta$. Therefore,
		\begin{align*}
		\FrN{\Tr_A\left[U(\sigma_1\otimes\sigma_2)U\ct\right]} &\leq \TrN{\Tr_A\left[U(\sigma_1\otimes\sigma_2)U\ct\right]} \\
		& = \TrN{\Tr_A\left[U(\sigma_1\otimes\sigma_2)U\ct\right] - \Tr_A\left[U_1\sigma_1U_1\ct\otimes U_2\sigma_2U_2\ct\right]} \\
		& \leq \TrN{U(\sigma_1\otimes\sigma_2)U\ct - (U_1\otimes U_2)(\sigma_1\otimes\sigma_2)(U_1\otimes U_2)\ct} \\
		&\leq 2\delta. \qedhere
		\end{align*}
	\end{proof}
	
	\begin{lemma}\label{lemma:sound}
		For $U\in\ugrp(4)$, if for every $\sigma_1\in\{X,Y,Z\}$ and $\sigma_2\in\{I,X,Y,Z\}$ we have
		\[\FrN{\Tr_A\left[U(\sigma_1\otimes\sigma_2)U\ct\right]}\leq \delta,\]
		then $d_\otimes(U,I\otimes I)\leq 20\sqrt{\delta}$.
	\end{lemma}
	
	The proof of \Cref{lemma:sound} is rather technical and thus is deferred to the end of this section. For now, let us show how \Cref{lemma:complete,lemma:sound} would imply the completeness and the soundness of \Cref{alg:learngate}.
	
	We apply \Cref{thm:mixing} to the circuit $\mc C$ minus its first layer, which is a circuit of depth $D-1$. Notice that among the input qubits to the second layer of $\mc C$, the third to $n$-th qubits only depends on the gates $G_{1,3},\ldots,G_{1,n}$ and thus can be viewed as fixed. The first qubit does not affect the $(D+1)$-th output qubit and thus its entire lightcone can be removed from the picture. That leaves us the second qubit, which is
	\[\Tr_A\left[G_{1,1}G\ct(\rho_1\otimes\rho_2)G G_{1,1}\ct\right].\]
	Therefore, with $\rho_1$ changed with a difference $\sigma_1$ and $\rho_2=(I+\sigma_2)/2$ unchanged, \Cref{thm:mixing} implies that with probability $1-\gamma$, the following holds for every $G$.
	\begin{equation}\label{eq:sound}
	\FrN{\pi-\pi'}\geq \FrN{\Tr_A\left[G_{1,1}G\ct(\sigma_1\otimes\rho_2)G G_{1,1}\ct\right]}\cdot (2^{-D+1}\gamma)^{c_2}.
	\end{equation}
	Meanwhile, \Cref{lemma:frob} implies that
	\begin{equation}\label{eq:complete}
	\FrN{\pi-\pi'}\leq \FrN{\Tr_A\left[G_{1,1}G\ct(\sigma_1\otimes\rho_2)G G_{1,1}\ct\right]}\cdot \sqrt{2}.
	\end{equation}
	
	For completeness, since we take $G$ from an $\eps$-net, the algorithm must have tested some $G$ with $d_\otimes(G_{1,1},G)=d_\otimes(G_{1,1}G\ct,I\otimes I)\leq\eps\leq\delta$. Hence by \Cref{lemma:complete} and \cref{eq:complete}, for such $G$ it always holds that $\FrN{\pi-\pi'}\leq2\sqrt{2}\eps$, and thus $G$ will not be rejected even with $\eps$ tomography errors in $\pi$ and $\pi'$.
	
	For soundness, assume that $d_\otimes(G_{1,1},G)>\delta$, then by \Cref{lemma:sound} we know that for some $\sigma_1\in\{X,Y,Z\}$ and $\sigma_2\in\{I,X,Y,Z\}$,
	\begin{equation}
	\FrN{\Tr_A\left[G_{1,1}G\ct(\sigma_1\otimes\sigma_2)G G_{1,1}\ct\right]}>\frac{1}{400}\delta^2.
	\end{equation}
	As $\sigma_2 = 2\rho_2-I$, that means there exists some $\rho_2$ such that
	\begin{equation}
	\FrN{\Tr_A\left[G_{1,1}G\ct(\sigma_1\otimes\rho_2)G G_{1,1}\ct\right]}>\frac{1}{1200}\delta^2.
	\end{equation}
	Thus by \cref{eq:sound} we have $\FrN{\pi-\pi'}>\delta^2(2^{-D+1}\gamma)^{c_2}/1200 > 7\eps$, and $G$ will be rejected when $\pi$ and $\pi'$ are estimated with $\eps$ error.
	
	The efficiency of the algorithm depends on the size of the $\eps$-net and the state tomography process, which are both $\poly(1/\eps)=\poly(1/\delta,1/\gamma)$. Notice that the algorithm similarly works for every gate in the first layer, and as a result, we obtained the following formal statement of \Cref{thm:learngate}:
	
	\begin{theorem}\label{thm:learngate2}
		Let $\mc C$ be a brickwork random quantum circuit of depth $D$, where each gate is independently Haar random. Let $G_{1,1}$ be the gate in the first layer of $\mc C$ that acts on the first and second qubit. Given oracle access to $\mc C$, for any $\delta,\gamma\in(0,1)$, with probability at least $1-\gamma$ over $\mc C$, \Cref{alg:learngate} outputs some $G\in\ugrp(4)$ such that $d_\otimes(G,G_{1,1})\leq \delta$ in time $\poly(2^D,1/\delta,1/\gamma)$.
	\end{theorem}
	
	The rest of this section is devoted to prove \Cref{lemma:sound}.
	
	\begin{proof}[Proof of \Cref{lemma:sound}]
		In the proof we assume that $\delta\leq 1/2$, as otherwise the claim is trivial.
		
		For any Hermitian operator $\sigma\in\C^{4\times 4}$ with trace zero, the value $\FrN{\Tr_A[\sigma]}/\sqrt{2}$ is exactly the length of the projection of $\sigma$ on to the subspace $I\otimes\spn\{X,Y,Z\}$. Therefore, the assumptions would imply that for every $\sigma\in\spn \{X,Y,Z\}\otimes \{I,X,Y,Z\}$ and $\sigma'\in I\otimes\spn\{X,Y,Z\}$, we have
		\begin{equation}\label{eq:s1}
		\Tr[U\sigma U\ct\sigma']\leq\frac{\delta}{\sqrt{2}}\cdot \FrN{\sigma}\FrN{\sigma'}.
		\end{equation}
		Now consider any single-qubit state $\rho$, and let $\sigma$ be the projection of $U\ct(I\otimes\rho)U$ onto the subspace $\spn \{X,Y,Z\}\otimes \{I,X,Y,Z\}$. Inequality \cref{eq:s1} gives
		\begin{align*}
		\Tr[U\ct(I\otimes\rho)U\sigma] &=\Tr[U\sigma U\ct(I\otimes\rho)] \nonumber\\
		& = \Tr[U\sigma U\ct(I\otimes(\rho-I/2))] \nonumber\\
		&\leq \frac{\delta}{\sqrt{2}}\cdot \FrN{\sigma}\cdot\sqrt{2}\FrN{\rho-I/2} \ \leq \delta.
		\end{align*}
		That means the projection of $U\ct(I\otimes\rho)U$ onto the orthogonal subspace $I\otimes\spn\{I,X,Y,Z\}$ must be large. Since this projection is exactly $\frac{1}{2}I\otimes\Tr_A[U\ct(I\otimes\rho)U]$, we have
		\begin{equation}\label{eq:s2}
		\frac{1}{2}\FrN{\Tr_A\left[U\ct(I\otimes\rho)U\right]}^2\geq \FrN{U\ct(I\otimes\rho)U}^2-\delta = 2\FrN{\rho}^2-\delta.
		\end{equation}
		
		If we define the following channel
		\[\Upphi(\rho) = \frac{1}{2}\Tr_A\left[U\ct(I\otimes\rho)U\right]\]
		then inequality \cref{eq:s2} translates to 
		\begin{equation}
		\FrN{\Upphi(\rho)}^2\geq \FrN{\rho}^2-\delta/2.
		\end{equation}
		That means $\Upphi$ is almost norm preserving and thus almost a unitary channel. In fact, as $\Upphi(I)=I$ by definition, $\Upphi$ is a unital channel and has a canonical form \cite{Choi23} $\Upphi'(\rho)=W\Upphi(V\rho V\ct)W\ct$ for some $V,W\in\ugrp(2)$ such that
		\[\Upphi'(X)=d_x X,\quad \Upphi'(Y)=d_y Y,\quad \Upphi'(Z)=d_z Z,\quad d_x,d_y,d_z\in[-1,1].\]
		By taking $\rho=(I+X)/2$ we get $d_x^2\geq 1-\delta$, and the same also holds for $d_y$ and $d_z$. The canonical form has an additional property that $(d_x,d_y,d_z)$ is a convex combination of vectors $(1,1,1),(1,-1,-1),(-1,1,-1),(-1,-1,1)$, and as $\delta\leq1/2$, $\Upphi'$ must be close to a Pauli rotation in $\{I,X,Y,Z\}$. That means there exists $U_2\in\{VW,VXW,VYW,VZW\}$ such that for every single-qubit state $\rho$,
		\begin{equation}
		\FrN{\Upphi(\rho)-U_2\ct\rho U_2}^2\leq \frac{1}{2}\left(1-\sqrt{1-\delta}\right)^2\leq \frac{1}{2}\delta^2.
		\end{equation}
		
		Therefore, for any two single-qubit states $\rho$ and $\rho'$ with $\Tr[\rho\rho']=0$ we have
		\begin{align}
		\Tr\left[U_2\Upphi(\rho)U_2\ct\rho'\right] &= \Tr\left[\left(U_2\Upphi(\rho)U_2\ct-\rho\right)\rho'\right] \nonumber\\
		&\leq \FrN{U_2\Upphi(\rho)U_2\ct-\rho}\FrN{\rho'} \nonumber\\
		&= \FrN{\Upphi(\rho)-U_2\ct\rho U_2}\FrN{\rho'} \leq \frac{1}{\sqrt{2}}\delta. \label{eq:s}
		\end{align}
		Now consider the unitary $U'=U(I\otimes U_2\ct)\in\ugrp(4)$, and we denote the entries of $U'$ as $u_{ij}$ for $i,j=1,\ldots,4$. Notice that
		\[U_2\Upphi(\rho)U_2\ct = \frac{1}{2}\Tr_A\left[(I\otimes U_2)U\ct(I\otimes\rho)U(I\otimes U_2\ct)\right] = \frac{1}{2}\Tr_A\left[{U'}\ct(I\otimes\rho)U'\right],\]
		and this allows us to write out $\Tr\left[U_2\Upphi(\rho)U_2\ct\rho'\right]$ exactly. In particular, when $\rho=\ket{0}\bra{0}$ and $\rho'=\ket{1}\bra{1}$, we get from \cref{eq:s} that
		\begin{equation}\label{eq:s3}
		|u_{12}|^2+|u_{14}|^2+|u_{32}|^2+|u_{34}|^2 \leq \sqrt{2}\delta.
		\end{equation}
		When $\rho=\ket{1}\bra{1}$ and $\rho'=\ket{0}\bra{0}$, we get
		\begin{equation}\label{eq:s4}
		|u_{21}|^2+|u_{23}|^2+|u_{41}|^2+|u_{43}|^2 \leq \sqrt{2}\delta.
		\end{equation}
		And when $\rho=\ket{+}\bra{+}$ and $\rho'=\ket{-}\bra{-}$, we get
		\begin{align}
		& |u_{11}+u_{21}-u_{12}-u_{22}|^2+|u_{13}+u_{23}-u_{14}-u_{24}|^2 \nonumber\\
		+\ & |u_{31}+u_{41}-u_{32}-u_{42}|^2+|u_{33}+u_{43}-u_{34}-u_{44}|^2 \leq 4\sqrt{2}\delta. \label{eq:s5}
		\end{align}
		Combine \cref{eq:s3,eq:s4,eq:s5} together we also get
		\begin{equation}\label{eq:s6}
		|u_{11}-u_{22}|^2+|u_{13}-u_{24}|^2+|u_{31}-u_{42}|^2+|u_{33}-u_{44}|^2 \leq 16\sqrt{2}\delta.
		\end{equation}
		Inequalities \cref{eq:s3,eq:s4,eq:s6} imply that there exists a matrix $M\in\C^{2\times 2}$ such that
		\begin{equation}
		\FrN{U'-M\otimes I}^2\leq 10\sqrt{2}\delta.
		\end{equation}
		
		Let $M=V'\Sigma W'$ be the singular value decomposition of $M$, then we have
		\begin{align}
		\FrN{U'-V'W'\otimes I} &\leq \FrN{U'-M\otimes I}+\FrN{V'W'\otimes I-M\otimes I} \nonumber\\
		&=\FrN{U'-M\otimes I} +\min_{U''\in\ugrp(4)}\FrN{U''-\Sigma\otimes I}\nonumber\\
		&\leq 2\FrN{U'-M\otimes I} \ \leq 10\sqrt{\delta}.
		\end{align}
		Therefore, if we let $U_1=V'W'$, then with the bound of diamond norm \cite{Haah23} we get
		\begin{align*}
		d_\otimes(U,I\otimes I)&\leq d_\diamond(U,U_1\otimes U_2) \nonumber \\
		&\leq 2\FrN{U - U_1\otimes U_2} = 2\FrN{U' - U_1\otimes I} \leq 20\sqrt{\delta}. \qedhere
		\end{align*}
	\end{proof}
	
	\subsubsection{Learning the Circuit with Discretized Distribution}\label{sec:disc}
	
	It is tempting to use \Cref{thm:learngate2} to learn the entire circuit $\mc C$. Indeed, if the statement is errorless that $d_\otimes(G,G_{1,1})=0$, we could view the single-qubit gates $G_{1,1}G\ct=U_1\otimes U_2$ as part of the second layer. That means after learning the first layer, we can perfectly uncompute it and hence use \Cref{alg:learngate} to learn the second layer, and proceed until we are only left with single qubit gates, which are easily learnable.
	
	However, when there are learning errors, which are inevitable for a continuous gate distribution like $\haar(4)$, the above framework runs into a problem. Since we cannot perfectly uncompute the first layer, the inputs to the rest of the circuit are not clean enough: In particular, the error in $\sigma_1\otimes\rho_2$ could have much larger influence on the output difference $\pi-\pi'$ than $\sigma_1$ itself. To control the influence of the errors, we need to reduce the learn error in the previous layer to be polynomially smaller than the target error in the current layer, and therefore a circuit of depth $D=\Uptheta(\log n)$ would require error as small as $2^{-\Omega(\log^2 n)}$ and incur quasi-polynomial running time.
	
	Here we present a bypass to the problem which allows us to prove an errorless version of \Cref{thm:learngate2} and thus make the proposed framework work. The idea is to change the distribution from $\ugrp(4)$ into a discrete one that approximates $\ugrp(4)$. Intuitively, any $\eps$-net where the elements are distributed according to $\haar(4)$ would be a good approximation. We formalize this intuition as the following:
	
	\begin{definition}\label{def:eps}
		An $\eps$-net of a distribution $\mc D$ over a pseudometric space $(S,d)$ is a distribution $\mc D_\eps$ over $S$ with a finite support, such that $\mc D_\eps=f(\mc D)$ for some (possibly randomized) map $f:S\to S$, with the following properties:
		\begin{itemize}
			\item For every $x\in S$, $d(x,f(x))\leq\eps$.
			\item For every $x_1,x_2\in S$, either $f(x_1)=f(x_2)$ or $d(f(x_1),f(x_2))\geq \eps$.
		\end{itemize}
	\end{definition}
	
	Notice that under \Cref{def:eps}, the set $\supp\mc D_\eps$ is indeed an $\eps$-net in the normal sense. Actually, the definition is general enough so that we can first choose any $\eps$-net as the support, remove the redundant elements with zero distances, and then take $f$ to be an arbitrary rounding scheme into the support. We show that $\mc D_\eps$ approximates $\mc D$ via the following lemma.
	
	\begin{lemma}\label{lemma:lip}
		If $F:S\to\R$ is $L$-Lipschitz, that is for all $x_1,x_2\in S$,
		\[|F(x_1)-F(x_2)|\leq L\cdot d(x_1,x_2),\]
		then for every $\delta\in\R$ we have
		\[\Pr_{x\sim\mc D_\eps}[F(x)\leq\delta]\leq \Pr_{x\sim\mc D}[F(x)\leq \delta+\eps L].\]
	\end{lemma}
	\begin{proof}
		Let $f:S\to S$ be the map in \Cref{def:eps}, then
		\begin{align*}
		\Pr_{x\sim\mc D_\eps}[F(x)\leq\delta] &= \Pr_{x\sim\mc D}[F(f(x))\leq\delta] \\
		&\leq \Pr_{x\sim\mc D}[F(x)\leq\delta +|F(x)-F(f(x))|] \\
		&\leq \Pr_{x\sim\mc D}[F(x)\leq\delta + L\cdot d(x,f(x))]\\ & \leq \Pr_{x\sim\mc D}[F(x)\leq \delta+\eps L]. \qedhere
		\end{align*}
	\end{proof}
	
	From now on, for each $\eps>0$ we fix some $\eps$-net of the Haar measure $\haar(4)$ under the $d_\otimes$ distance, and denote it by $\haar_\eps(4)$. We will show that when the gates in the brickwork circuit are drawn the net, we can actually use the framework at the start of this section to learn the circuit. To do so, we first prove an errorless version of \Cref{thm:learngate2} as follows.
	
	\begin{theorem}\label{thm:learngate3}
		Let $\mc C$ be a brickwork random quantum circuit of depth $D$, where each gate is independently drawn from $\haar_\eps(4)$. Let $G_{1,1}$ be the gate in the first layer of $\mc C$ that acts on the first and second qubit. Given oracle access to $\mc C$, for every $\gamma\in(0,1)$ there is an algorithm that with probability at least $1-\gamma$ over $\mc C$ outputs $G_{1,1}$ in time $\poly(2^D,1/\gamma)$.
	\end{theorem}
	\begin{proof}
		The algorithm is basically the same as \Cref{alg:learngate}, except that we now iterate $G$ through the support of $\haar_\eps(4)$. The proof is also mostly the same: When $G=G_{1,1}$, we have $\pi=\pi'$ and thus $G$ will not be rejected; Otherwise $d_\otimes (G,G_{1,1})\geq \eps$, and the proof goes through as long as we have the corresponding version of \Cref{thm:mixing}.
		
		Since \Cref{thm:mixing} is proved via \Cref{lemma:mixing}, it suffices to prove \Cref{lemma:mixing} where $\haar(4)$ is replaced with $\haar_\eps(4)$. The crux is to prove the inequality \cref{eq:m2}, that is for some $C,C'>0$, it holds for all $x\geq 0$ that
		\begin{equation}\label{eq:l1}
		\Pr[\left|\det M\right|\leq x]\leq C'x^C,
		\end{equation}
		where $M$ is the matrix form of the linear map $M:\tau-\tau'\mapsto \Tr_A[G(\tau-\tau')G\ct]$, for $\tau-\tau'$ in a certain fixed 3-dimensional subspace of $\Delta_2$.
		
		Note that by \Cref{prop:bound}, each entry of $M$ is in $[-2,2]$, while by \cref{eq:deg}, each entry of $M$ is a Lipschitz function of $G$ under distance $d_\diamond$. This is because when $G,G'\in\ugrp(4)$ that $d_\diamond(G,G')\leq\delta$ corresponds to matrices $M$ and $M'$, for any $\tau-\tau'\in\Delta_2$ and $\sigma\in\Delta_1$ with $\FrN{\tau-\tau'}=\FrN{\sigma}=1$ we have
		\begin{align}
		\left|\Tr[\sigma M(\tau-\tau')]-\Tr[\sigma M'(\tau-\tau')]\right|
		&=\left|\Tr\left[(I\otimes \sigma)\left(G(\tau-\tau') G\ct - G'(\tau-\tau') {G'}\ct\right)\right]\right| \nonumber\\
		&\leq\TrN{G(\tau-\tau') G\ct - G'(\tau-\tau') {G'}\ct} \nonumber\\
		&\leq \TrN{G\tau G\ct - G'\tau {G'}\ct} + \TrN{G\tau' G\ct - G'\tau' {G'}\ct} \nonumber\\
		&\leq 2\delta.
		\end{align}
		As $\det M$ consists of $6$ monomials of degree $3$ in the entries of $M$, we conclude that $\left|\det M\right|$ is $2^3\cdot 3\cdot 6=144$-Lipschitz in $G$ under $d_\diamond$. But when $G'=(U_1\otimes U_2)G$ we have $M'=U_2MU_2\ct$ which means $\left|\det M\right|=\left|\det M'\right|$, and thus $\left|\det M\right|$ is also 144-Lipschitz in $G$ under $d_\otimes$.
		
		As a result, since we already know that \cref{eq:l1} holds when $G\sim \haar(4)$, by \Cref{lemma:lip} we have
		\begin{equation}
		\Pr_{G\sim \haar_\eps(4)}[\left|\det M\right|\leq x] \leq \Pr_{G\sim \haar(4)}[\left|\det M\right|\leq x+144\eps] \leq C'(x+144\eps)^C.
		\end{equation}
		However, this will not yield \Cref{lemma:mixing} as the bound is even non-zero when $x=0$. Fortunately, we can simply use the union bound to ignore the cases when $\left|\det M\right|\leq \eps$ for any gate $G$ on the path, which will only add $\poly(\eps)D$ to the error probability $\gamma$. And conditioned on $\left|\det M\right|>\eps$, for every $x\geq 0$ we have
		\begin{equation}
		\Pr_{G\sim \haar_\eps(4)}[\left|\det M\right|\leq x] \leq C'(145 x)^C,
		\end{equation}
		which allows us to prove \Cref{lemma:mixing} on $\haar_\eps(4)$, albeit with different constants.
	\end{proof}
	
	As a result, we can use \Cref{thm:learngate3} to exactly learn the circuit $\mc C$ layer by layer, and obtain \Cref{thm:learncirc} formally as follows.
	
	\begin{theorem}\label{thm:learncirc2}
		Let $\mc C$ be a brickwork random quantum circuit on $n$ qubits of depth $D$, where each gate is independently drawn from $\haar_\eps(4)$. Given oracle access to $\mc C$, for every $\gamma\in(0,1)$ there is an algorithm that with probability at least $1-\gamma$ outputs $\mc C$ in time $\poly(n,2^D,1/\gamma)$.
	\end{theorem}
		

\section*{Acknowledgments}
 B.F., S.G., and W.Z. ~acknowledge support from AFOSR (FA9550-21-1-0008). The authors thank Anurag Anshu, Adam Bouland, Lijie Chen,  Jonas Haferkamp, Robert Huang, Issac Kim, Yunchao Liu, Tony Metger, Marcus Michelen, Tommy Schuster and Kewen Wu for helpful comments and discussions, and thank Jacob Watkins for pointing out a mistake in our previous version of \Cref{sec:disc}. This material is based upon work partially supported by the National Science Foundation under Grant CCF-2044923 (CAREER), by the U.S. Department of Energy, Office of Science, National Quantum Information Science Research Centers (Q-NEXT) and by the DOE QuantISED grant DE-SC0020360.  
	\printbibliography

@misc{hunterjones2019unitarydesignsstatisticalmechanics,
      title={Unitary designs from statistical mechanics in random quantum circuits}, 
      author={Nicholas Hunter-Jones},
      year={2019},
      eprint={1905.12053},
      archivePrefix={arXiv},
      primaryClass={quant-ph}
}

@misc{chen2024efficientunitarydesignspseudorandom,
      title={Efficient unitary designs and pseudorandom unitaries from permutations}, 
      author={Chi-Fang Chen and Adam Bouland and Fernando G. S. L. Brandão and Jordan Docter and Patrick Hayden and Michelle Xu},
      year={2024},
      eprint={2404.16751},
      archivePrefix={arXiv},
      primaryClass={quant-ph}
}

@article{Haferkamp22,
	doi = {10.22331/q-2022-09-08-795},
	title = {Random quantum circuits are approximate unitary {$t$}-designs in depth {$O\left(nt^{5+o(1)}\right)$}},
	author = {Haferkamp, Jonas},
	journal = {{Quantum}},
	issn = {2521-327X},
	publisher = {{Verein zur F{\"{o}}rderung des Open Access Publizierens in den Quantenwissenschaften}},
	volume = {6},
	pages = {795},
	month = Sep,
	year = {2022}
}

@article{PhysRevLett.124.180601,
  title = {Operator L\'evy Flight: Light Cones in Chaotic Long-Range Interacting Systems},
  author = {Zhou, Tianci and Xu, Shenglong and Chen, Xiao and Guo, Andrew and Swingle, Brian},
  journal = {Phys. Rev. Lett.},
  volume = {124},
  issue = {18},
  pages = {180601},
  numpages = {6},
  year = {2020},
  month = May,
  publisher = {American Physical Society},
  doi = {10.1103/PhysRevLett.124.180601}
}

@article{Chen_2021,
   title={Operator Growth Bounds from Graph Theory},
   volume={385},
   ISSN={1432-0916},
   DOI={10.1007/s00220-021-04151-6},
   number={3},
   journal={Communications in Mathematical Physics},
   publisher={Springer Science and Business Media LLC},
   author={Chen, Chi-Fang and Lucas, Andrew},
   year={2021},
   month=jul, pages={1273–1323} }

@article{Chen_2023,
   title={Speed limits and locality in many-body quantum dynamics},
   volume={86},
   ISSN={1361-6633},
   DOI={10.1088/1361-6633/acfaae},
   number={11},
   journal={Reports on Progress in Physics},
   publisher={IOP Publishing},
   author={Chen, Chi-Fang and Lucas, Andrew and Yin, Chao},
   year={2023},
   month=sep, pages={116001} }

@misc{jian2022lineargrowthcircuitcomplexity,
      title={Linear Growth of Circuit Complexity from Brownian Dynamics}, 
      author={Shao-Kai Jian and Gregory Bentsen and Brian Swingle},
      year={2022},
      eprint={2206.14205},
      archivePrefix={arXiv},
      primaryClass={quant-ph}
}

@misc{metger2024simpleconstructionslineardepthtdesigns,
      title={Simple constructions of linear-depth t-designs and pseudorandom unitaries}, 
      author={Tony Metger and Alexander Poremba and Makrand Sinha and Henry Yuen},
      year={2024},
      eprint={2404.12647},
      archivePrefix={arXiv},
      primaryClass={quant-ph}
}

@article{Dalzell22,
   title={Random Quantum Circuits Anticoncentrate in Log Depth},
   volume={3},
   ISSN={2691-3399},
   DOI={10.1103/prxquantum.3.010333},
   number={1},
   journal={PRX Quantum},
   publisher={American Physical Society (APS)},
   author={Dalzell, Alexander M. and Hunter-Jones, Nicholas and Brandão, Fernando G. S. L.},
   year={2022},
   month=mar }

@article{Nahum17,
  title = {Quantum Entanglement Growth under Random Unitary Dynamics},
  volume = {7},
  ISSN = {2160-3308},
  DOI = {10.1103/physrevx.7.031016},
  number = {3},
  journal = {Physical Review X},
  publisher = {American Physical Society (APS)},
  author = {Nahum,  Adam and Ruhman,  Jonathan and Vijay,  Sagar and Haah,  Jeongwan},
  year = {2017},
  month = jul 
}

@misc{haah2024efficientapproximateunitarydesigns,
      title={Efficient approximate unitary designs from random Pauli rotations}, 
      author={Jeongwan Haah and Yunchao Liu and Xinyu Tan},
      year={2024},
      eprint={2402.05239},
      archivePrefix={arXiv},
      primaryClass={quant-ph}
}

@misc{chen2024efficientunitarytdesignsrandom,
      title={Efficient Unitary T-designs from Random Sums}, 
      author={Chi-Fang Chen and Jordan Docter and Michelle Xu and Adam Bouland and Patrick Hayden},
      year={2024},
      eprint={2402.09335},
      archivePrefix={arXiv},
      primaryClass={quant-ph}
}

@article{Boixo_2018,
   title={Characterizing quantum supremacy in near-term devices},
   volume={14},
   ISSN={1745-2481},
   DOI={10.1038/s41567-018-0124-x},
   number={6},
   journal={Nature Physics},
   publisher={Springer Science and Business Media LLC},
   author={Boixo, Sergio and Isakov, Sergei V. and Smelyanskiy, Vadim N. and Babbush, Ryan and Ding, Nan and Jiang, Zhang and Bremner, Michael J. and Martinis, John M. and Neven, Hartmut},
   year={2018},
   month=apr, pages={595–600} }

@misc{Yang23,
      title = {The Complexity of Learning (Pseudo)random Dynamics of Black Holes and Other Chaotic Systems}, 
      author = {Lisa Yang and Netta Engelhardt},
      year = {2023},
      eprint = {2302.11013},
      archivePrefix = {arXiv},
      primaryClass = {hep-th}
}

@article{Piroli20,
   title = {A random unitary circuit model for black hole evaporation},
   volume = {2020},
   ISSN = {1029-8479},
   DOI = {10.1007/jhep04(2020)063},
   number = {4},
   journal = {Journal of High Energy Physics},
   publisher = {Springer Science and Business Media LLC},
   author = {Piroli, Lorenzo and Sünderhauf, Christoph and Qi, Xiao-Liang},
   year = {2020},
   month = Apr
}

@inproceedings{Aaronson23,
    author = {Aaronson, Scott and Hung, Shih-Han},
    title = {Certified Randomness from Quantum Supremacy},
    year = {2023},
    isbn = {9781450399135},
    publisher = {Association for Computing Machinery},
    address = {New York, NY, USA},
    doi = {10.1145/3564246.3585145},
    booktitle = {Proceedings of the 55th Annual ACM Symposium on Theory of Computing},
    pages = {933–944},
    numpages = {12},
    location = {Orlando, FL, USA},
    series = {STOC 2023}
}

@misc{Bassirian24,
      title = {On Certified Randomness from Fourier Sampling or Random Circuit Sampling}, 
      author = {Roozbeh Bassirian and Adam Bouland and Bill Fefferman and Sam Gunn and Avishay Tal},
      year = {2024},
      eprint = {2111.14846},
      archivePrefix = {arXiv},
      primaryClass = {quant-ph}
}

@misc{Liu22,
      title = {Benchmarking near-term quantum computers via random circuit sampling}, 
      author = {Yunchao Liu and Matthew Otten and Roozbeh Bassirianjahromi and Liang Jiang and Bill Fefferman},
      year = {2022},
      eprint = {2105.05232},
      archivePrefix = {arXiv},
      primaryClass = {quant-ph}
}

@misc{Fefferman23,
      title = {Effect of non-unital noise on random circuit sampling}, 
      author = {Bill Fefferman and Soumik Ghosh and Michael Gullans and Kohdai Kuroiwa and Kunal Sharma},
      year = {2023},
      eprint = {2306.16659},
      archivePrefix = {arXiv},
      primaryClass = {quant-ph}
}

@inproceedings{Bouland22,
   title = {Noise and the Frontier of Quantum Supremacy},
   DOI = {10.1109/focs52979.2021.00127},
   booktitle = {2021 IEEE 62nd Annual Symposium on Foundations of Computer Science (FOCS)},
   publisher = {IEEE},
   author = {Bouland, Adam and Fefferman, Bill and Landau, Zeph and Liu, Yunchao},
   year = {2022},
   month = Feb
}

@article{Bouland18,
   title = {On the complexity and verification of quantum random circuit sampling},
   volume = {15},
   ISSN = {1745-2481},
   DOI = {10.1038/s41567-018-0318-2},
   number = {2},
   journal = {Nature Physics},
   publisher = {Springer Science and Business Media LLC},
   author = {Bouland, Adam and Fefferman, Bill and Nirkhe, Chinmay and Vazirani, Umesh},
   year = {2018},
   month = Oct,
    pages = {159–163}
}

@misc{Morvan23,
      title={Phase transition in Random Circuit Sampling}, 
      author={A. Morvan and B. Villalonga and X. Mi and S. Mandrà and A. Bengtsson and P. V. Klimov and Z. Chen and S. Hong and C. Erickson and I. K. Drozdov and J. Chau and G. Laun and R. Movassagh and A. Asfaw and L. T. A. N. Brandão and R. Peralta and D. Abanin and R. Acharya and R. Allen and T. I. Andersen and K. Anderson and M. Ansmann and F. Arute and K. Arya and J. Atalaya and J. C. Bardin and A. Bilmes and G. Bortoli and A. Bourassa and J. Bovaird and L. Brill and M. Broughton and B. B. Buckley and D. A. Buell and T. Burger and B. Burkett and N. Bushnell and J. Campero and H. S. Chang and B. Chiaro and D. Chik and C. Chou and J. Cogan and R. Collins and P. Conner and W. Courtney and A. L. Crook and B. Curtin and D. M. Debroy and A. Del Toro Barba and S. Demura and A. Di Paolo and A. Dunsworth and L. Faoro and E. Farhi and R. Fatemi and V. S. Ferreira and L. Flores Burgos and E. Forati and A. G. Fowler and B. Foxen and G. Garcia and E. Genois and W. Giang and C. Gidney and D. Gilboa and M. Giustina and R. Gosula and A. Grajales Dau and J. A. Gross and S. Habegger and M. C. Hamilton and M. Hansen and M. P. Harrigan and S. D. Harrington and P. Heu and M. R. Hoffmann and T. Huang and A. Huff and W. J. Huggins and L. B. Ioffe and S. V. Isakov and J. Iveland and E. Jeffrey and Z. Jiang and C. Jones and P. Juhas and D. Kafri and T. Khattar and M. Khezri and M. Kieferová and S. Kim and A. Kitaev and A. R. Klots and A. N. Korotkov and F. Kostritsa and J. M. Kreikebaum and D. Landhuis and P. Laptev and K. -M. Lau and L. Laws and J. Lee and K. W. Lee and Y. D. Lensky and B. J. Lester and A. T. Lill and W. Liu and W. P. Livingston and A. Locharla and F. D. Malone and O. Martin and S. Martin and J. R. McClean and M. McEwen and K. C. Miao and A. Mieszala and S. Montazeri and W. Mruczkiewicz and O. Naaman and M. Neeley and C. Neill and A. Nersisyan and M. Newman and J. H. Ng and A. Nguyen and M. Nguyen and M. Yuezhen Niu and T. E. O'Brien and S. Omonije and A. Opremcak and A. Petukhov and R. Potter and L. P. Pryadko and C. Quintana and D. M. Rhodes and E. Rosenberg and C. Rocque and P. Roushan and N. C. Rubin and N. Saei and D. Sank and K. Sankaragomathi and K. J. Satzinger and H. F. Schurkus and C. Schuster and M. J. Shearn and A. Shorter and N. Shutty and V. Shvarts and V. Sivak and J. Skruzny and W. C. Smith and R. D. Somma and G. Sterling and D. Strain and M. Szalay and D. Thor and A. Torres and G. Vidal and C. Vollgraff Heidweiller and T. White and B. W. K. Woo and C. Xing and Z. J. Yao and P. Yeh and J. Yoo and G. Young and A. Zalcman and Y. Zhang and N. Zhu and N. Zobrist and E. G. Rieffel and R. Biswas and R. Babbush and D. Bacon and J. Hilton and E. Lucero and H. Neven and A. Megrant and J. Kelly and I. Aleiner and V. Smelyanskiy and K. Kechedzhi and Y. Chen and S. Boixo},
      year={2023},
      eprint={2304.11119},
      archivePrefix={arXiv},
      primaryClass={quant-ph}
}

@article{Arute2019,
  title = {Quantum supremacy using a programmable superconducting processor},
  volume = {574},
  ISSN = {1476-4687},
  DOI = {10.1038/s41586-019-1666-5},
  number = {7779},
  journal = {Nature},
  publisher = {Springer Science and Business Media LLC},
  author = {Arute,  Frank and Arya,  Kunal and Babbush,  Ryan and Bacon,  Dave and Bardin,  Joseph C. and Barends,  Rami and Biswas,  Rupak and Boixo,  Sergio and Brandao,  Fernando G. S. L. and Buell,  David A. and Burkett,  Brian and Chen,  Yu and Chen,  Zijun and Chiaro,  Ben and Collins,  Roberto and Courtney,  William and Dunsworth,  Andrew and Farhi,  Edward and Foxen,  Brooks and Fowler,  Austin and Gidney,  Craig and Giustina,  Marissa and Graff,  Rob and Guerin,  Keith and Habegger,  Steve and Harrigan,  Matthew P. and Hartmann,  Michael J. and Ho,  Alan and Hoffmann,  Markus and Huang,  Trent and Humble,  Travis S. and Isakov,  Sergei V. and Jeffrey,  Evan and Jiang,  Zhang and Kafri,  Dvir and Kechedzhi,  Kostyantyn and Kelly,  Julian and Klimov,  Paul V. and Knysh,  Sergey and Korotkov,  Alexander and Kostritsa,  Fedor and Landhuis,  David and Lindmark,  Mike and Lucero,  Erik and Lyakh,  Dmitry and Mandrà,  Salvatore and McClean,  Jarrod R. and McEwen,  Matthew and Megrant,  Anthony and Mi,  Xiao and Michielsen,  Kristel and Mohseni,  Masoud and Mutus,  Josh and Naaman,  Ofer and Neeley,  Matthew and Neill,  Charles and Niu,  Murphy Yuezhen and Ostby,  Eric and Petukhov,  Andre and Platt,  John C. and Quintana,  Chris and Rieffel,  Eleanor G. and Roushan,  Pedram and Rubin,  Nicholas C. and Sank,  Daniel and Satzinger,  Kevin J. and Smelyanskiy,  Vadim and Sung,  Kevin J. and Trevithick,  Matthew D. and Vainsencher,  Amit and Villalonga,  Benjamin and White,  Theodore and Yao,  Z. Jamie and Yeh,  Ping and Zalcman,  Adam and Neven,  Hartmut and Martinis,  John M.},
  year = {2019},
  month = oct,
  pages = {505–510}
}

@book{Meckes19,
    title = {The random matrix theory of the classical compact groups},
    author = {Meckes, Elizabeth S},
    volume = {218},
    year = {2019},
    publisher = {Cambridge University Press},
    doi = {10.1017/9781108303453}
}

@article{Fassino19,
    author = {Fassino, Claudia and Pistone, Giovanni and Rogantin, Maria Piera},
    title = {Computing the Moments of the Complex {G}aussian: Full and Sparse Covariance Matrix},
    journal = {Mathematics},
    volume = {7},
    year = {2019},
    number = {3},
    article-number = {263},
    issn = {2227-7390},
    doi = {10.3390/math7030263}
}

@inproceedings{Haah23,
  title = {Query-optimal estimation of unitary channels in diamond distance},
  author = {Haah, Jeongwan and Kothari, Robin and O'Donnell, Ryan and Tang, Ewin},
  booktitle = {2023 IEEE 64th Annual Symposium on Foundations of Computer Science (FOCS)},
  pages = {363--390},
  year = {2023},
  organization = {IEEE},
  doi = {10.1109/FOCS57990.2023.00028}
}

@article{Brandao12,
	title = {Local Random Quantum Circuits are Approximate Polynomial-Designs},
 	author = {Brand{\~a}o, Fernando G.S.L. and Harrow, Aram W. and Horodecki, Micha\l{}},
	journal = {Commun. Math. Phys.},
	volume = {346},
	number = {2},
	pages = {397--434},
	year = {2016},
	doi = {10.1007/s00220-016-2706-8}
}

@article{Harrow23,
	author = {Harrow, Aram W. and Mehraban, Saeed},
	title = {Approximate Unitary {$t$}-Designs by Short Random Quantum Circuits Using Nearest-Neighbor and Long-Range Gates},
	journal = {Commun. Math. Phys.},
	volume = {401},
	number = {2},
	pages = {1531--1626},
	year = {2023},
	doi = {10.1007/s00220-023-04675-z}
}

@article{Bertini20,
	title = {Scrambling in random unitary circuits: Exact results},
	author = {Bertini, Bruno and Piroli, Lorenzo},
	journal = {Phys. Rev. B},
	volume = {102},
	issue = {6},
	pages = {064305},
	numpages = {25},
	year = {2020},
	month = Aug,
	publisher = {American Physical Society},
	doi = {10.1103/PhysRevB.102.064305}
}

@misc{Brown13,
	title = {Scrambling speed of random quantum circuits}, 
	author = {Winton Brown and Omar Fawzi},
	year = {2013},
	eprint = {1210.6644},
	archivePrefix = {arXiv},
	primaryClass = {quant-ph}
}

@article{Dankert09,
	title = {Exact and approximate unitary 2-designs and their application to fidelity estimation},
	author = {Dankert, Christoph and Cleve, Richard and Emerson, Joseph and Livine, Etera},
	journal = {Phys. Rev. A},
	volume = {80},
	issue = {1},
	pages = {012304},
	numpages = {6},
	year = {2009},
	month = Jul,
	publisher = {American Physical Society},
	doi = {10.1103/PhysRevA.80.012304}
}

@article{Emerson05,
	author = {Emerson, Joseph and Alicki, Robert and Zyczkowski, Karol},
	year = {2005},
	month = {03},
	pages = {},
	title = {Scalable Noise Estimation with Random Unitary Operators},
	volume = {7},
	journal = {Journal of Optics B: Quantum and Semiclassical Optics},
	doi = {10.1088/1464-4266/7/10/021}
}

@article{Choi23,
	title = {On unital qubit channels},
	author = {Choi, Man-Duen and Li, Chi-Kwong},
	journal = {Quantum Inf. Comput.},
	volume = {23},
	number = {7-8},
	pages = {562--576},
	year = {2023},
	doi = {10.26421/QIC23.7-8-2}
}

@article{Hangleiter24,
	title = {Bell Sampling from Quantum Circuits},
	author = {Hangleiter, Dominik and Gullans, Michael J.},
	journal = {Phys. Rev. Lett.},
	volume = {133},
	issue = {2},
	pages = {020601},
	numpages = {7},
	year = {2024},
	month = Jul,
	publisher = {American Physical Society},
	doi = {10.1103/PhysRevLett.133.020601}
}

@article{Carbery01,
	doi = {10.4310/mrl.2001.v8.n3.a1},
	year = {2001},
	publisher = {International Press of Boston},
	volume = {8},
	number = {3},
	pages = {233--248},
	author = {Anthony Carbery and James Wright},
	title = {Distributional and {$L^q$} norm inequalities for polynomials over convex bodies in {$\mathbb{R}^n$}},
	journal = {Mathematical Research Letters}
}

@inproceedings{Huang24,
	author = {Huang, Hsin-Yuan and Liu, Yunchao and Broughton, Michael and Kim, Isaac and Anshu, Anurag and Landau, Zeph and McClean, Jarrod R.},
	title = {Learning Shallow Quantum Circuits},
	year = {2024},
	publisher = {Association for Computing Machinery},
	address = {New York, NY, USA},
	doi = {10.1145/3618260.3649722},
	booktitle = {Proceedings of the 56th Annual ACM Symposium on Theory of Computing},
	pages = {1343-1351},
	numpages = {9},
	location = {Vancouver, BC, Canada},
	series = {STOC 2024}
}

@misc{Schuster24,
	title = {Random unitaries in extremely low depth}, 
	author = {Thomas Schuster and Jonas Haferkamp and Hsin-Yuan Huang},
	year = {2024},
	eprint = {2407.07754},
	archivePrefix = {arXiv},
	primaryClass = {quant-ph}
}

@article{Hangleiter18,
	doi = {10.22331/q-2018-05-22-65},
	title = {Anticoncentration theorems for schemes showing a quantum speedup},
	author = {Hangleiter, Dominik and Bermejo{-}Vega, Juan and Schwarz, Martin and Eisert, Jens},
	journal = {{Quantum}},
	issn = {2521-327X},
	publisher = {{Verein zur F{\"{o}}rderung des Open Access Publizierens in den Quantenwissenschaften}},
	volume = {2},
	pages = {65},
	month = may,
	year = {2018}
}

@article{Fisher23,
	title={Random quantum circuits},
	author={Fisher, Matthew P.A. and Khemani, Vedika and Nahum, Adam and Vijay, Sagar},
	journal={Annual Review of Condensed Matter Physics},
	volume={14},
	number={1},
	pages={335--379},
	year={2023},
	publisher={Annual Reviews},
	issn = {1947-5462},
	doi = {10.1146/annurev-conmatphys-031720-030658}
}

@misc{Braccia24,
	title={Computing exact moments of local random quantum circuits via tensor networks}, 
	author={Paolo Braccia and Pablo Bermejo and Lukasz Cincio and M. Cerezo},
	year={2024},
	eprint={2403.01706},
	archivePrefix={arXiv},
	primaryClass={quant-ph}
}

@article{Meka16,
	author = {Meka, Raghu and Nguyen, Oanh and Vu, Van},
	title = {Anti-concentration for Polynomials of Independent Random Variables},
	year = {2016},
	pages = {1--17},
	doi = {10.4086/toc.2016.v012a011},
	publisher = {Theory of Computing},
	journal = {Theory of Computing},
	volume = {12},
	number = {11}
}

@article{Kane17,
	author = {Daniel Kane},
	title = {A structure theorem for poorly anticoncentrated polynomials of Gaussians and applications to the study of polynomial threshold functions},
	volume = {45},
	journal = {The Annals of Probability},
	number = {3},
	publisher = {Institute of Mathematical Statistics},
	pages = {1612 -- 1679},
	year = {2017},
	doi = {10.1214/16-AOP1097}
}

@inproceedings{Aaronson11,
	author = {Aaronson, Scott and Arkhipov, Alex},
	title = {The computational complexity of linear optics},
	year = {2011},
	isbn = {9781450306911},
	publisher = {Association for Computing Machinery},
	address = {New York, NY, USA},
	doi = {10.1145/1993636.1993682},
	pages = {333-342},
	numpages = {10},
	location = {San Jose, California, USA},
	series = {STOC '11}
}

@misc{Lovett10,
	title = {An elementary proof of anti-concentration of polynomials in Gaussian variables},
	author = {Lovett, Shachar},
	howpublished = {Electronic Colloquium on Computational Complexity (ECCC)},
	year = {2010},
	url = {https://eccc.weizmann.ac.il/report/2010/182}
}

@article{Meka13,
	author = {Meka, Raghu and Zuckerman, David},
	title = {Pseudorandom Generators for Polynomial Threshold Functions},
	journal = {SIAM Journal on Computing},
	volume = {42},
	number = {3},
	pages = {1275-1301},
	year = {2013},
	doi = {10.1137/100811623}
}

@article{Erdos61,
	title = {On a classical problem of probability theory},
	author = {Erd{\H{o}}s Paul and R{\'e}nyi, Alfr{\'e}d},
	journal = {Magyar Tud. Akad. Mat. Kutat{\'o} Int. K{\"o}zl},
	volume = {6},
	number = {1},
	pages = {215--220},
	year = {1961}
}

@article{Harrow21,
	title = {Separation of Out-Of-Time-Ordered Correlation and Entanglement},
	author = {Harrow, Aram W. and Kong, Linghang and Liu, Zi-Wen and Mehraban, Saeed and Shor, Peter W.},
	journal = {PRX Quantum},
	volume = {2},
	issue = {2},
	pages = {020339},
	numpages = {12},
	year = {2021},
	month = Jun,
	publisher = {American Physical Society},
	doi = {10.1103/PRXQuantum.2.020339}
}

@article{Nahum18,
	title = {Operator Spreading in Random Unitary Circuits},
	author = {Nahum, Adam and Vijay, Sagar and Haah, Jeongwan},
	journal = {Phys. Rev. X},
	volume = {8},
	issue = {2},
	pages = {021014},
	numpages = {30},
	year = {2018},
	month = Apr,
	publisher = {American Physical Society},
	doi = {10.1103/PhysRevX.8.021014}
}

@article{Harrow09,
	title = {Random Quantum Circuits are Approximate 2-designs},
	author = {Harrow, Aram W. and Low, Richard A.},
	volume = {291},
	ISSN = {1432-0916},
	DOI = {10.1007/s00220-009-0873-6},
	number = {1},
	journal = {Communications in Mathematical Physics},
	publisher = {Springer Science and Business Media LLC},
	year = {2009},
	month = Jul,
	pages={257–302}
}

@article{Movassagh23,
    author = {Movassagh, Ramis},
    title = {The hardness of random quantum circuits},
    journal = {Nature Physics},
    year = {2023},
    month = Nov,
    day = {01},
    volume = {19},
    number = {11},
    pages = {1719--1724},
    issn = {1745-2481},
    doi = {10.1038/s41567-023-02131-2}
}

@misc{Landau24,
      title={Learning quantum states prepared by shallow circuits in polynomial time}, 
      author={Zeph Landau and Yunchao Liu},
      year={2024},
      eprint={2410.23618},
      archivePrefix={arXiv},
      primaryClass={quant-ph}
}

@article{Hayden07,
   title={Black holes as mirrors: quantum information in random subsystems},
   volume={2007},
   ISSN={1029-8479},
   DOI={10.1088/1126-6708/2007/09/120},
   number={09},
   journal={Journal of High Energy Physics},
   publisher={Springer Science and Business Media LLC},
   author={Hayden, Patrick and Preskill, John},
   year={2007},
   month=sep, pages={120–120} }

@misc{Chen24Incomp,
      title={Incompressibility and spectral gaps of random circuits}, 
      author={Chi-Fang Chen and Jeongwan Haah and Jonas Haferkamp and Yunchao Liu and Tony Metger and Xinyu Tan},
      year={2024},
      eprint={2406.07478},
      archivePrefix={arXiv},
      primaryClass={quant-ph}
}

@article{Lieb72,
author = "Lieb, {Elliott H.} and Robinson, {Derek W.}",
year = "1972",
month = sep,
doi = "10.1007/BF01645779",
language = "English (US)",
volume = "28",
pages = "251--257",
journal = "Communications In Mathematical Physics",
issn = "0010-3616",
publisher = "Springer New York",
number = "3",
}

@article{Wilming22,
  title = {Lieb-Robinson bounds imply locality of interactions},
  author = {Wilming, Henrik and Werner, Albert H.},
  journal = {Phys. Rev. B},
  volume = {105},
  issue = {12},
  pages = {125101},
  numpages = {11},
  year = {2022},
  month = Mar,
  publisher = {American Physical Society},
  doi = {10.1103/PhysRevB.105.125101}
}

@book{Nielsen_Chuang_2010,
place={Cambridge}, 
title={Quantum Computation and Quantum Information: 10th Anniversary Edition}, 
publisher={Cambridge University Press}, 
author={Nielsen, Michael A. and Chuang, Isaac L.}, 
year={2010}}
	
\end{document}